\RequirePackage{fix-cm}

\makeatletter
\disable@package@load{program}{}
\makeatother
\documentclass[twocolumn]{svjour3}          %
\smartqed  %
\usepackage{url}
\usepackage{braket}
\usepackage{physics}
\usepackage{graphicx}
\usepackage{bm}
\usepackage{xcolor}
\usepackage{algorithm}
\usepackage{algpseudocode}
\usepackage{mathtools}
\usepackage{amsmath,amssymb}
\numberwithin{equation}{section}
\usepackage{subfigure}
\usepackage{afterpage}
\usepackage{comment}
\usepackage{qcircuit}
\usepackage{bookmark}
\usepackage{breakcites}
\usepackage{natbib}
\usepackage{cuted}
\usepackage{caption}
\usepackage{ulem}
\usepackage{float}
\usepackage{placeins}
\providecommand{\doi}[1]{DOI~\href{http://dx.doi.org/#1}{\path|#1|}}

\usepackage{multirow}

\journalname{Quantum Machine Intelligence}

\algnewcommand\algorithmicforeach{\textbf{for each}}
\algdef{S}[FOR]{ForEach}[1]{\algorithmicforeach\ #1\ \algorithmicdo}

\newcommand\independent{\protect\mathpalette{\protect\independenT}{\perp}}
\def\independenT#1#2{\mathrel{\rlap{$#1#2$}\mkern2mu{#1#2}}}
\newcommand\notindependent{\not\!\independent}

\usepackage{theoremref}
\usepackage{thmtools}

\usepackage{braket}

\begin{document}
\sloppy

\title{Quantum-enhanced causal discovery for a small number of samples}

\author{
Yu Terada$\mbox{}^{\dagger}$ \and
Ken Arai$\mbox{}^{\dagger}$ \and
Yu Tanaka$\mbox{}^{\dagger}$ \and
Yota Maeda$\mbox{}^{\dagger}$ \and
Hiroshi Ueno$\mbox{}^{}$ \and
Hiroyuki Tezuka$\mbox{}^{\star}$
}

\institute{
       Advanced Research Laboratory, Sony Group Corporation, 1-7-1 Konan, Minato-ku, Tokyo, 108-0075, Japan \\
       $\star$ Corresponding author: hiroyuki.tezuka@sony.com \\
       $\dagger$ These authors contributed equally to this work.
}

\maketitle

\begin{abstract}

The discovery of causal relations from observed data has attracted significant interest from disciplines such as economics, social sciences, epidemiology, and biology. 
In practical applications, considerable knowledge of the underlying systems is often unavailable, and real data are usually associated with nonlinear causal structures, which makes the direct use of most conventional causality analysis methods difficult.
This study proposes a novel quantum Peter-Clark (qPC) algorithm for causal discovery that does not require any assumptions about the underlying model structures.
Based on conditional independence tests in a class of reproducing kernel Hilbert spaces characterized by quantum circuits, the proposed \textit{qPC} algorithm can explore causal relations from the observed data drawn from arbitrary distributions.
We conducted extensive and systematic experiments on fundamental graph parts of causal structures, demonstrating that the qPC algorithm exhibits significantly better performance, particularly with smaller sample sizes compared to its classical counterpart.
Furthermore, we proposed a novel optimization approach based on Kernel Target Alignment (KTA) for determining hyperparameters of quantum kernels. 
This method effectively reduced the risk of false positives in causal discovery, enabling more reliable inference.
Our theoretical and experimental results demonstrate that the proposed quantum algorithm can empower classical algorithms for robust and accurate inference in causal discovery, supporting them in regimes where classical algorithms typically fail.
In addition, the effectiveness of this method was validated using the datasets on Boston housing prices, heart disease, and biological signaling systems as real-world applications. 
These findings highlight the potential of quantum circuit-based causal discovery methods in addressing practical challenges, particularly in small-sample scenarios, where traditional approaches have shown significant limitations.
\end{abstract}


\section{Introduction}
\label{section:intro}
Deciphering causal relations among observed variables is a crucial problem in the social and natural sciences.
Historically, interventions or randomized experiments have been used as standard approaches to assess causality among observed variables (\cite{pearl2018book}).
For example, randomized controlled trials have been commonly used in clinical research to assess the potential effects of drugs.
However, conducting such interventions or randomized experiments is often challenging due to ethical constraints and high costs.
Alternatively, causal discovery provides practical methods for inferring causal relations between variables from observed data, extending beyond correlation analysis (\cite{spirtes2001causation,glymour2019review,vowels2022d,camps2023discovering,hasan2023survey}).
The Peter--Clark (PC) algorithm~(\cite{spirtes2001causation}), a widely accepted algorithm for causal discovery, yields an equivalence class of directed acyclic graphs (DAGs) that captures causal relations (see Fig.~\ref{figure:qpc} (a) for an overview of the PC algorithm).
The PC algorithm does not assume any specific statistical models or data distributions, unlike the other methods, including the linear non-Gaussian acyclic model (LiNGAM)~(\cite{10.5555/1248547.1248619, 10.5555/1953048.2021040}), NOTEARs~(\cite{zheng2018dags}), the additive noise model~(\cite{hoyer2008nonlinear}), the post nonlinear causal model~(\cite{zhang2012identifiability}), and the Greedy Equivalence Search (GES) algorithm~(\cite{chickering2002optimal}).
Thus, applications of the PC algorithm and its variants have elucidated causal relations from various observed data spanning from natural science to engineering (\cite{le2013inferring,runge2019inferring,nowack2020causal,castri2023enhancing}).
In the PC algorithm, kernel methods can be used for conditional independent tests, a process known as kernel-based conditional independence test (KCIT) (\cite{10.5555/3020548.3020641, zhang2012kernelbasedconditionalindependencetest}). This approach enables applications for various types of data, including those characterized by nonlinearity and high dimensionality (\cite{zhang2012kernelbasedconditionalindependencetest,strobl2019approximate,runge2019inferring}).

Although the PC algorithm using KCIT can be applied to both linear and nonlinear data without making any assumptions about the underlying models, its performance depends on the choice of kernels.
Empirically, kernels are often chosen from representative classes such as Gaussian, polynomial, and linear kernels (\cite{zheng2024causal}).
Alternatively, quantum models that embed data in an associated reproducing kernel Hilbert space (RKHS) have recently been developed, providing a class of algorithms called quantum kernel methods~(\cite{schuld2021supervised,jerbi2023quantum,thanasilp2024exponential,glick2024covariant,kawaguchi2023application}) (see an example of quantum circuits in Fig.~\ref{figure:qpc} (b)). 
Among them, the kernel-based LiNGAM extended with quantum kernels~(\cite{kawaguchi2023application}) demonstrates potential advantages over classical methods, such as accurate inference with small sample sizes~(\cite{maeda2023estimation}), as suggested in supervised learning contexts~(\cite{caro2022generalization}).
However, the quantum LiNGAM (qLiNGAM) (\cite{kawaguchi2023application}) assumes linear causal relations, which limits its applicability to real-world problems.

\begin{figure*}[t]
  \includegraphics[width=0.95\textwidth]{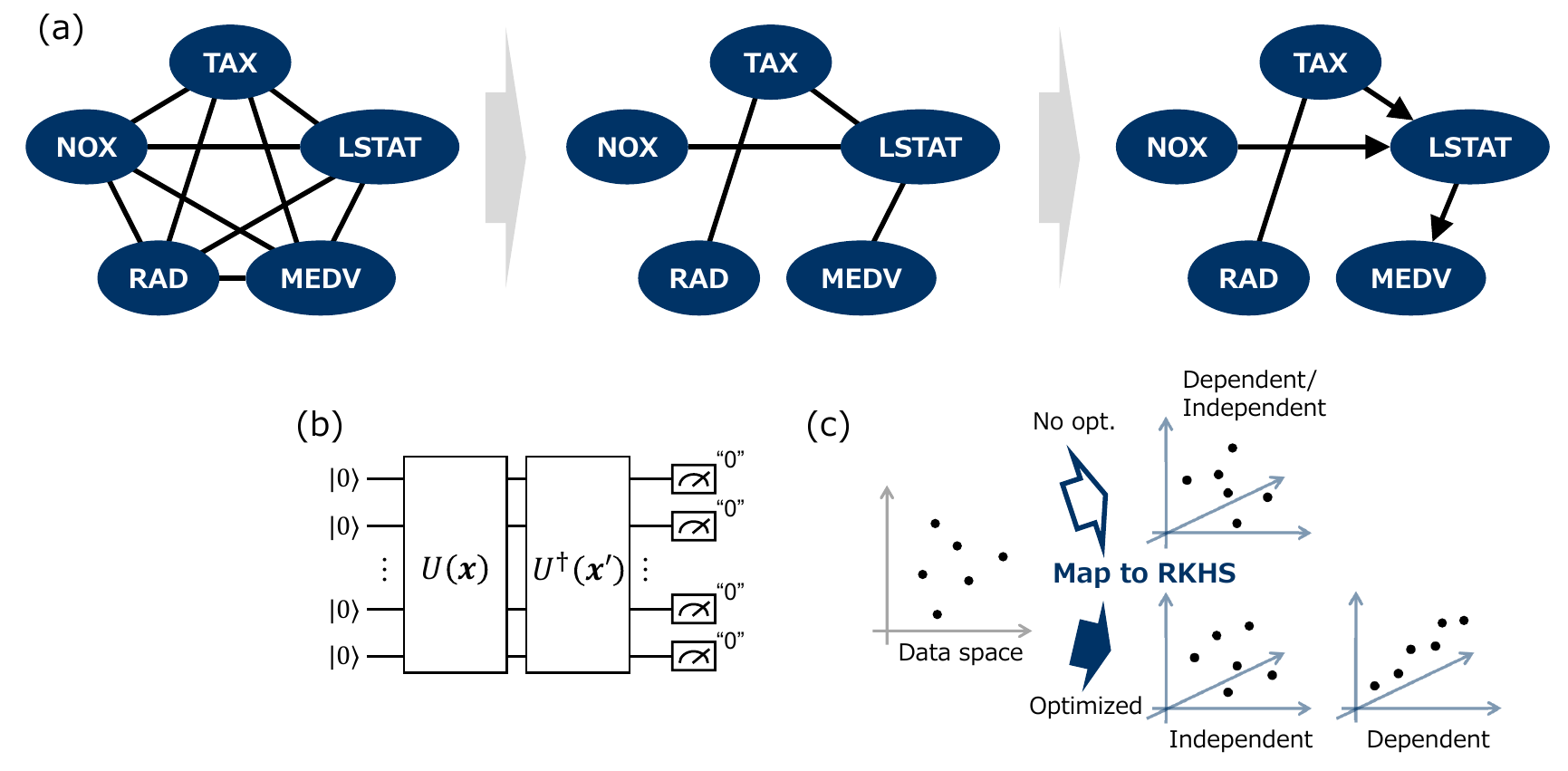}
  \caption{Schematic of the proposed quantum Peter--Clerk (qPC) algorithm and our optimization method based on kernel target alignment (KTA). 
  (a) Overview of the qPC algorithm. 
  Left: The graph representation of an initial input. The qPC algorithm identifies causal relations among random variables and represents them as complete, partially directed acyclic graphs (CPDAGs). The qPC algorithm begins with a complete undirected graph, where each node represents a random variable, and each edge represents the correlation between two random variables. 
  The middle: The graph of the (conditional) independence among the random variables. The algorithm prunes redundant edges by performing the (conditional) independence test between two random variables conditioned on other random variables. Note that when performing the conditional independence test between any two random variables, the set of random variables used for conditioning is recorded. 
  Right: The resulting causal graph. The edges can be oriented following the rules (the details are described in Appendix \ref{app:pc}). 
  (b) Quantum circuit for a kernel. We defined the kernel, $k(x, y)$, for the KCIT as the inner product of quantum states $U_{\theta}(\mathbf{x})\ket{0}^{\otimes n}$ and $U_{\theta}(\mathbf{x'})\ket{0}^{\otimes n}$ generated from the parameterized unitary $U_{\theta}$. 
  (c) Overview of kernel optimization for independence test in causal discovery. If an inappropriate and non-optimized kernel is used for the independence test, it fails to detect the dependent or independent relation between variables accurately. The optimized kernel can disentangle complex relations between variables, allowing for the accurate discrimination of dependent or independent relations in statistical tests.} 
  \label{figure:qpc}
\end{figure*}

Quantum-enhanced causal inference and discovery for small-sample data show promise but face challenges.
First, existing quantum models have failed to address nonlinear causal relations.
Second, similar to classical kernels, the performance of quantum kernel methods depends critically on the choice of quantum circuits used~(\cite{shaydulin2022importance}), and systematic approaches for selecting appropriate quantum kernels in causal discovery are still lacking. 
In most previous studies that employed classical methods, kernel parameters, such as the median strategy, were often selected heuristically \cite{zheng2024causal}. 
Moreover, no established methods exist for setting the hyperparameters of quantum circuits.
Finally, it remains unclear why causal inference using quantum kernels outperforms classical methods for small sample data.

To address these challenges, we propose the quantum PC (qPC) algorithm, which leverages the quantum kernel in the independence tests of the PC algorithm (Fig.~\ref{figure:qpc}).  
We then propose a novel method based on \textit{kernel target alignment (KTA)}~\cite{NIPS2001_1f71e393} to determine the appropriate hyperparameters in quantum kernels for causal discovery.
The proposed method enables the setting of kernels with objective criteria and eliminates arbitrariness in kernel method applications.  
Furthermore, we discuss how the qPC algorithm can enhance inference accuracy in small sample sizes.
Using KTA, we demonstrate that the quantum models we used can effectively learn to produce kernels with high independence detection capabilities.
To demonstrate that our optimization method based on the KTA facilitates accurate causal discovery by the qPC algorithm through the selection of appropriate kernels, we used simulations based on three-node causal graphs (Fig.~\ref{figure:solo_sumo}(a)), which are the fundamental blocks of general causal graphs.

To validate the practical effectiveness of the qPC algorithm, we conducted comprehensive evaluations using both quantum and classical data sources. 
Our first simulation, motivated by the superiority of quantum kernels in small-sample regimes, employs quantum circuit models to generate data from which causal discovery methods infer the underlying causal relations.
While the data from quantum models can highlight the characteristics of the qPC algorithm, it is desirable to use classical data to estimate the typical performance of the quantum method using the proposed kernel choice process in practical applications.
Thus, we assessed the situations in which we observed data drawn from classical systems.
The optimization method based on the KTA bridges the gap between the qPC algorithm and realistic data.
Using the proposed kernel choice method, we demonstrate the applicability of the qPC algorithm to real and synthetic data.
The real data include those from the Boston housing price (\cite{harrison1978hedonic}) and clinical observations related to heart disease (\cite{ahmad2017survival}), and biological signaling systems (\cite{sachs2005causal}.
The results obtained by the qPC algorithm provide insights that align with domain knowledge, which classical methods cannot, and highlight the usefulness of the quantum method for small datasets.

\section{qPC algorithm}
\label{section:qpc}

\subsection{Overview of the qPC algorithm}
\label{sec:results_qpc}

We propose the qPC algorithm for causal discovery, which employs quantum kernel methods (\cite{schuld2021supervised}) to embed classical data into quantum states (Fig. \ref{figure:qpc} (c)).
The qPC algorithm is an extension of the PC algorithm for causal inference. 
It utilizes a conditional independence test implemented via the KCIT with quantum kernels composed of data-embedded quantum states as a natural extension of the Gaussian kernel.

The original PC algorithm~(\cite{spirtes1991algorithm,spirtes2001causation}) offers CPDAGs that capture the causal relations between variables from their observed data (Appendix \ref{app:pc}).
This algorithm is a nonparametric method that does not consider underlying statistical models.
The KCIT is introduced because of its powerful capacity to infer causality in data with nonlinearity and high dimensionality (\cite{10.5555/3020548.3020641, zhang2012kernelbasedconditionalindependencetest}).

Specifically, the qPC algorithm involves two main steps: determining unconditional and conditional independence among variables and orienting causality relations (see the overview of the PC algorithm in Appendix~\ref{app:pc}).
The qPC algorithm outputs CPDAGs, which capture the causal relations among the observed variables, featuring both directed and undirected edges between them (Fig. \ref{figure:qpc} (a)).
It relies on the KCIT framework (see Appendix~\ref{app:kcit} for the details of the KCIT), where the original data are embedded into feature spaces to detect independence (Fig. \ref{figure:qpc} (b)).
Appropriate embedding in KCIT facilitates the disentangling of complex nonlinear relations in the original data space, which often leads to accurate results in statistical hypothesis tests, especially when dealing with high-dimensional or nonlinear data~(\cite{10.5555/3020548.3020641, zhang2012kernelbasedconditionalindependencetest}).
The qPC algorithm leverages quantum kernels associated with the quantum state to embed data into the RKHS defined by quantum circuits.
Quantum kernels are defined by $k_Q(\mathbf{x}, \mathbf{x}')=\mathrm{Tr}[\rho(\mathbf{x})\rho(\mathbf{x}')]$, where input $\mathbf{x}$ is encoded into the quantum circuits generating state $\rho(\mathbf{x})$.
Our proposed quantum circuit has hyperparameters analogous to the widths of the Gaussian kernels.

\subsection{Details of the quantum kernel-based conditional tests for the qPC algorithm}\label{subsec:qkct}

The KCIT~(\cite{10.5555/3020548.3020641, zhang2012kernelbasedconditionalindependencetest}) is a hypothesis test for null hypothesis $X \perp \!\!\!\! \perp Y \ | \ Z$ between random variables $X$ and $Y$ given $Z$. 
It was developed as a conditional independence test by defining a simple statistic based on HSIP of two centralized conditional kernel matrices and deriving its asymptotic distribution under the null hypothesis (see Appendix \ref{app:kcit} for details).
Unconditional independence statistic $T_{UI}$ is defined as 
\begin{eqnarray}
    T_{UI} := \frac{1}{n} {\rm Tr} \bigl[ \widetilde{\mathbf{K}}_X \widetilde{\mathbf{K}}_Y \bigr],
\end{eqnarray}
where $\widetilde{\mathbf{K}}_X$ and $\widetilde{\mathbf{K}}_Y$ are the centralized kernel matrices \textit{i.i.d.} of size $n$ for $X$ and $Y$.
Under the null hypothesis that $X$ and $Y$ are statistically independent, it follows that the Gamma distribution
\begin{eqnarray}
    p(t) = t^{k-1} \frac{e^{-t/\theta}}{\theta^k \Gamma (k)},
\end{eqnarray}
where shape parameter $k$ and scale parameter $\theta$ are estimated by 
\begin{eqnarray}
    k &=& \frac{{\rm Tr} \bigr[ \widetilde{\mathbf{K}}_X \bigl]^2 {\rm Tr} \bigr[ \widetilde{\mathbf{K}}_Y \bigl]^2}{2 {\rm Tr} \bigr[ \widetilde{\mathbf{K}}_X^2 \bigl] {\rm Tr} \bigr[ \widetilde{\mathbf{K}}_Y^2 \bigl]}, \\
    \theta &=& \frac{2 {\rm Tr} \bigr[ \widetilde{\mathbf{K}}_X^2 \bigl] {\rm Tr} \bigr[ \widetilde{\mathbf{K}}_Y^2 \bigl]}{n^2 {\rm Tr} \bigr[ \widetilde{\mathbf{K}}_X \bigl] {\rm Tr} \bigr[ \widetilde{\mathbf{K}}_Y \bigl]}.
\end{eqnarray}
The conditional independence statistic, $T_{CI}$, is defined as
\begin{eqnarray}
    T_{CI} := \frac{1}{n} \mathrm{Tr} \bigl[ \widetilde{\mathbf{K}}_{\ddot{\mathbf{X}} | \mathbf{Z}} \widetilde{\mathbf{K}}_{\mathbf{Y} | \mathbf{Z}} \bigr],
\end{eqnarray}
where $\ddot{X} = (X, Z)$, $\widetilde{\mathbf{K}}_{\ddot{X}|Z} = \mathbf{R}_Z \widetilde{\mathbf{K}}_{\ddot{X}} \mathbf{R}_Z$ and 
$\mathbf{R}_Z = \mathbf{I} - \widetilde{\mathbf{K}}_Z (\widetilde{\mathbf{K}}_Z + \epsilon \mathbf{I})^{-1} = \epsilon (\widetilde{\mathbf{K}}_Z + \epsilon \mathbf{I})^{-1}$.
We constructed $\widetilde{\mathbf{K}}_{Y|Z}$ similarly.
Although $T_{CI}$ also approximately follows the gamma distribution under the null hypothesis, parameters $k$ and $\theta$ are described by a matrix based on the eigenvectors $\widetilde{\mathbf{K}}_{\ddot{X}|Z}$ and $\widetilde{\mathbf{K}}_{Y|Z}$.

We employed a quantum kernel to design the kernel matrices.
The most basic quantum kernel is calculated using the fidelity of two quantum states: the embedded data $\mathbf{x}$ and $\mathbf{x'}$, $k(\mathbf{x, x'})=\mathrm{Tr}[\rho(\mathbf{x})\rho(\mathbf{x'}) ]$~(\cite{havlivcek2019supervised}).
Data-embedded quantum states are generated using a parameterized quantum circuit.
As shown in Fig.~\ref{fig:qc_whole}, data $\mathbf{x}$ are mapped into the quantum state via the unitary operation as $U(\mathbf{x}) \ket{0}^{\otimes n} =\Pi^{n_{\mathrm{dep}}}_{i}U_i(\mathbf{x})U_{\mathrm{init}}\ket{0}^{\otimes n}$, where $n$ is the number of qubits and $n_{\mathrm{dep}}$ is the number of data reuploading.
This operation offers the effect of superposition and entanglement between qubits.
Here, if we design an appropriate quantum circuit, the data will be effectively mapped onto the RKHS suitable for the KCIT.
The details of the quantum circuits tested in this study are described in Appendix~\ref{app:qc}.
The key to designing an effective quantum circuit lies in selecting the components of the unitary operation and pre- and post-processing the data.
Pre-processing involves scaling and affine transformations of the embedding data, while post-processing entails designing the observables.
In this study, we introduced only scaling for pre-processing and employed fidelity as the observable parameter for simplicity.

\begin{center}\
\begin{figure*}[h!]
    \centering
    \includegraphics[width=0.7\textwidth]{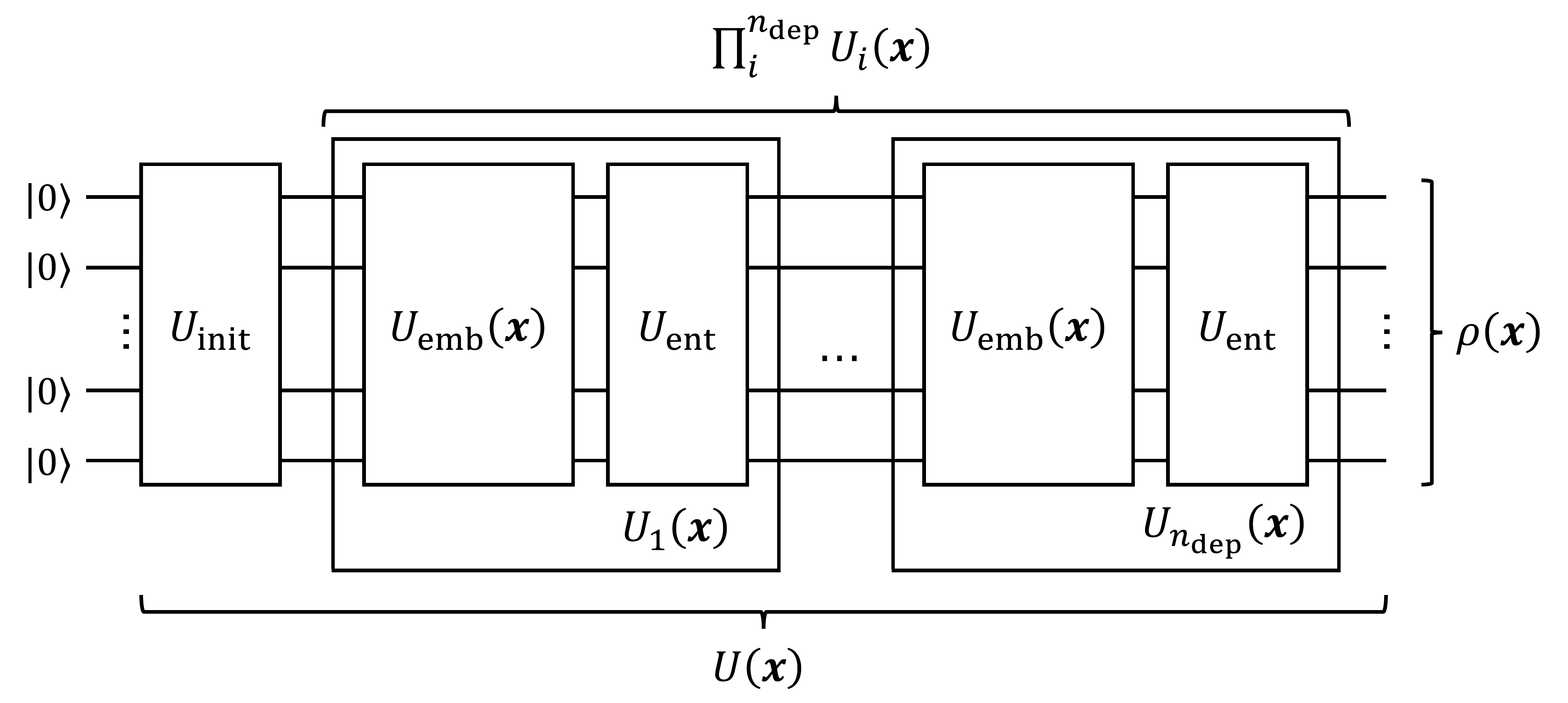}
    \caption{Structure of the quantum circuit for generating the quantum state.}   
    \label{fig:qc_whole}
\end{figure*}
\end{center}\

\section{Optimization of quantum circuits via KTA}\label{sec:results_optimization}

\subsection{Overview of quantum kernel optimization via KTA}

In the experimental section \ref{section:experiments}, we will first confirm that quantum kernels with small sample sizes are effective for causal discovery, where artificial data generated from quantum circuits, which are considered suitable for quantum kernels, are used. 
However, na\"ive quantum kernels are not suitable for classical data in general.
Specifically, the qPC algorithm has one main challenge: in contrast to the classical Gaussian kernel, which has several established guidelines for determining the kernel hyperparameters, the quantum kernel method lacks a standardized approach for selecting its hyperparameters for inference (\cite{shaydulin2022importance}).
Thus, we propose a systematic method for adjusting the hyperparameters in quantum circuits for datasets. 
To demonstrate the applicability of the qPC algorithm to a wide range of data, we compare the performance of the two methods using artificial datasets with classical settings.

Herein, we briefly explain an optimization method for determining the hyperparameters of quantum circuits for kernels based on the normalized Hilbert-Schmidt inner product (HSIP).
Its expectation value is zero if and only if random variables $X$ and $Y$ are independent.
This property enables the use of HSIP as test statistics in statistical hypothesis tests (\cite{10.5555/3020548.3020641, zhang2012kernelbasedconditionalindependencetest}).
The hypothesis test should be improved by selecting a kernel that minimizes the HSIP for uncorrelated data samples while maximizing the HSIP for correlated data samples; in principle, HSIP approaches zero in the uncorrelated case and is nonzero otherwise.
The normalized HSIP \eqref{eq:kta_loss}, which measures the distance between the feature vectors in which two data samples are embedded, is called KTA~(\cite{NIPS2001_1f71e393}).
From the perspective of statistical hypothesis testing, KTA minimization for uncorrelated data reduces the false-positive (FP) risk, whereas KTA maximization for correlated data reduces the false-negatives (FN) risk.
Thus, KTA minimization can be interpreted as enhancing the identifiability of two independent random variables, thereby reducing the likelihood of Type-I errors. 
In contrast, KTA maximization reduces the identifiability of dependent random variables, thereby decreasing the likelihood of Type-II errors.
Here, we focus on KTA minimization for uncorrelated data because the actual relations behind the data are often unavailable, making it challenging to employ the KTA maximization strategy.

\subsection{Details of kernel optimization via KTA}
We discuss kernel selection for the unconditional independence test and propose optimization heuristics based on KTA~(\cite{NIPS2001_1f71e393}) in more detail.
We rely on the fact that the statistics are extracted from the HSIP, which measures the discrepancy between feature vectors.
$X$ and $Y$ are independent if and only if the feature vectors of the embedded data in RKHS are orthogonal.
Intuitively, this leads to the selection of a kernel that minimizes (resp. maximizes) the HSIP for independent (resp. dependent) data samples.

We define the normalized HSIP {\it i.e.}, the KTA
\begin{eqnarray}
    {\rm KTA}(X, Y) := \frac{\mathrm{Tr} \bigl[ \widetilde{\mathbf{K}}_{\mathbf{X}} \widetilde{\mathbf{K}}_{\mathbf{Y}} \bigr]}{\sqrt{\mathrm{Tr} \bigl[ \widetilde{\mathbf{K}}_{\mathbf{X}}^2 \bigr] \mathrm{Tr} \bigl[ \widetilde{\mathbf{K}}_{\mathbf{Y}}^2 \bigr]}}, \label{eq:kta_loss}
\end{eqnarray}
as the evaluation function.
The normalized HSIP can be interpreted as the signal-to-noise ratio ${\rm S/N}$ of the asymptotic gamma distribution under the null hypothesis. This is demonstrated by Theorem~\ref{thm:gamma_param} (Proposition 5 of ref.~(\cite{10.5555/3020548.3020641, zhang2012kernelbasedconditionalindependencetest}) as follows:
\begin{eqnarray}
    {\rm S/N}
    &:=& 
    \frac{\mathbb{E}\left[ \breve{T}_{UI} \mid \mathcal{D} \right]}{\sqrt{\mathbb{V}ar\left[ \breve{T}_{UI} \mid \mathcal{D} \right]}} \\
    &=&
    \frac{\mathrm{Tr} \bigl[ \widetilde{\mathbf{K}}_{\mathbf{X}} \widetilde{\mathbf{K}}_{\mathbf{Y}} \bigr]}{\sqrt{\mathrm{Tr} \bigl[ \widetilde{\mathbf{K}}_{\mathbf{X}}^2 \bigr] \mathrm{Tr} \bigl[ \widetilde{\mathbf{K}}_{\mathbf{Y}}^2 \bigr]}} \\
    &=&
    {\rm KTA}(X, Y). 
    \label{eq:signal_to_noise_ratio}
\end{eqnarray}

The derivatives of Eq.~\eqref{eq:kta_loss} for minimization is expressed as follows:

\begin{lemma} \label{lem:deriv}
    For parameterized kernels $(\mathbf{K}_X)_{xx'} = k_X(x, x' | \theta)$ and $(\mathbf{K}_Y)_{yy'} =  k_Y(y, y' | \phi)$, consider the following function:
    \begin{eqnarray}
        f(\theta, \phi) &=& -\log \left( \frac{{\rm Tr}[\mathbf{K}_X \mathbf{K}_Y]}{\sqrt{{\rm Tr}[\mathbf{K}_X^2]{\rm Tr}[\mathbf{K}_Y^2]}} \right) \nonumber \\
        &=& -\log \left( {\rm KTA} \left(\mathbf{K}_X, \mathbf{K}_Y \right) \right).
    \end{eqnarray}
    The derivatives of the function are then given by
    \begin{eqnarray}
        \frac{\partial f}{\partial \theta} &=& 
        -\frac{{\rm Tr}[ (2 \mathbf{K}_Y - \mathbf{K}_Y \circ \mathbf{I}) \partial_{\theta} \mathbf{K}_X ]}{{\rm Tr}[ \mathbf{K}_X \mathbf{K}_Y ]} \nonumber \\
        &&+ \frac{{\rm Tr}[ (2 \mathbf{K}_X - \mathbf{K}_X \circ \mathbf{I}) \partial_{\theta} \mathbf{K}_X ]}{{\rm Tr}[ \mathbf{K}_X^2 ]}, \\
        \frac{\partial f}{\partial \phi} &=& 
        -\frac{{\rm Tr}[ (2 \mathbf{K}_X - \mathbf{K}_X \circ \mathbf{I}) \partial_{\phi} \mathbf{K}_Y ]}{{\rm Tr}[ \mathbf{K}_X \mathbf{K}_Y ]} \nonumber \\
        &&+ \frac{{\rm Tr}[ (2 \mathbf{K}_Y - \mathbf{K}_Y \circ \mathbf{I}) \partial_{\phi} \mathbf{K}_Y ]}{{\rm Tr}[ \mathbf{K}_Y^2 ]},
    \end{eqnarray}
    where $(\partial_{\theta} \mathbf{K}_X)_{xx'} = \partial_{\theta} k_X(x, x'|\theta)$ and $(\partial_{\theta} \mathbf{K}_Y)_{yy'} = \partial_{\phi} k_Y(y, y'|\phi)$.
\end{lemma}

\begin{proof}
    See Appendix~\ref{app:prf}.
\end{proof}

\subsection{Implementation of the kernel optimization}

We now explain the actual implementation of optimizing classical and quantum kernels.
As mentioned in the previous subsection, we minimize KTA in Eq. ~\eqref{eq:kta_loss} for the independent data samples. 
One natural method is to eliminate the correlation between two random variables by random shuffling of given data samples. We then minimize KTA using the gradient descent. 
The random shuffling method generates independent data while preserving the marginal distribution, and minimizing the KTA for such data reduces the signal-to-noise ratio in Eq. ~\eqref{eq:signal_to_noise_ratio} under the null hypothesis. 
From the perspective of statistical hypothesis testing, the KTA minimization reduces the false-positive (FP) risk.
We present the pseudocode for the gradient-based KTA minimization in Algorithm~\ref{alg:KTAMin}.

An alternative method is to sample the assumed marginal distributions in advance, whose moments are estimated using the given data samples. 
Sampling from modeled marginal distributions has the advantage of allowing the generation of large data samples, whereas the random shuffle method does not require prior knowledge of the marginal distribution.
In our experiments, we adopted the random shuffling method for small data samples.
To minimize the KTA, we employed a sampling-based method, such as branch and bound (\cite{Grund1979ForsytheGE, Brent,2020SciPy-NMeth}), rather than a differentiation-based method.

\begin{figure}[!t]
\begin{algorithm}[H]
\caption{KTA Minimization} \label{alg:KTAMin}
\begin{algorithmic}[1]
\Require Data samples $\mathcal{D}_{X,Y} = \{ (x_i, y_i) \}_{i=1}^n$, the target value $\epsilon > 0$, the difference parameter $\eta > 0$, and the sample number $m$.
\Ensure The parameters $(\theta, \phi)$ of ${\rm KTA}(X, Y)$ in Eq.~\eqref{eq:kta_loss}. 
\State {\bf [Initialization]} \label{alg:KTA1}
\State Calculate the means $m_X$, and $m_Y$ from the data samples $\mathcal{D}_{X,Y}$, respectively. \label{alg:KTA2}
\State Calculate the variances $\sigma_X^2$, and $\sigma_Y^2$ from $\mathcal{D}_{X,Y}$, respectively. \label{alg:KTA3}
\State $\theta = (\theta_1, ..., \theta_{|\theta|}) \sim \mathcal{N}(0, 1)$. \label{alg:KTA4}
\State $\phi = (\phi_1, ..., \phi_{|\theta|}) \sim \mathcal{N}(0, 1)$. \label{alg:KTA5}
\State Set a positive value larger than $\epsilon$ to $f(\theta, \phi) = -\log {\rm KTA}(X,Y)$. \label{alg:KTA6}
\State {\bf [Main loop]} \label{alg:KTA7}
\While{$f(\theta, \phi)$ is larger than $\epsilon$} \label{alg:KTA8}
    \State $X = (x_1, ..., x_m) \sim \mathcal{N}(m_X, \sigma_X)$. \label{alg:KTA9}
    \State $Y = (y_1, ..., y_m) \sim \mathcal{N}(m_Y, \sigma_Y)$. \label{alg:KTA10}
    \State Calculate the centralized kernel matrix $\widetilde{\mathbf{K}}_X$, and $\widetilde{\mathbf{K}}_Y$ from $(X,Y)$, respectively. \label{alg:KTA11}
    \State Calculate $\partial_{\theta} f = -{\rm Tr}[ (2 \widetilde{\mathbf{K}}_Y - \widetilde{\mathbf{K}}_Y \circ \mathbf{I}) \partial_{\theta} \widetilde{\mathbf{K}}_X ] / {\rm Tr}[ \widetilde{\mathbf{K}}_X \widetilde{\mathbf{K}}_Y ] + {\rm Tr}[ (2 \widetilde{\mathbf{K}}_X - \widetilde{\mathbf{K}}_X \circ \mathbf{I}) \partial_{\theta} \widetilde{\mathbf{K}}_X ] / {\rm Tr}[ \widetilde{\mathbf{K}}_X^2 ]$. \label{alg:KTA12}
    \State Calculate $\partial_{\phi} = -{\rm Tr}[ (2 \widetilde{\mathbf{K}}_X - \widetilde{\mathbf{K}}_X \circ \mathbf{I}) \partial_{\phi} \widetilde{\mathbf{K}}_Y ] / {\rm Tr}[ \widetilde{\mathbf{K}}_X \widetilde{\mathbf{K}}_Y ] + {\rm Tr}[ (2 \widetilde{\mathbf{K}}_Y - \widetilde{\mathbf{K}}_Y \circ \mathbf{I}) \partial_{\phi} \widetilde{\mathbf{K}}_Y ] / {\rm Tr}[ \widetilde{\mathbf{K}}_Y^2 ]$. \label{alg:KTA13}
    \State $\theta \leftarrow \theta + \eta \partial_{\theta} f$. \label{alg:KTA14}
    \State $\phi \leftarrow \phi + \eta \partial_{\phi} f$. \label{alg:KTA15}
    \State Calculate and update $f(\theta, \phi)$. \label{alg:KTA16}
\EndWhile \label{alg:KTA17}
\end{algorithmic}
\end{algorithm}
\end{figure}

\section{Experiments}
\label{section:experiments}

\subsection{Detection of fundamental causal graph structures}\label{sec:results_solosumo}
To demonstrate how the qPC algorithm can effectively retrieve the underlying causal structures, we applied it to synthetic data from fundamental causal relations with three nodes, collider, fork, chain, and independent structures (Fig. \ref{figure:solo_sumo} (a)) (\cite{pearl2018book}).
These elements capture any local part of the general causal graphs, thereby providing a summarized assessment of causal discovery methods.
In particular, we assume that source random variables are generated through observations in quantum circuits with random variable inputs and that the other nodes receive their inputs through a relation defined by the function $f$ and the external noise $\epsilon$, such as $Z =f(X, Y)+\epsilon$ (Fig. \ref{figure:solo_sumo} (b)).
Specifically, random values $\mathbf{x}$ sampled from the Gaussian distributions were used as inputs to the data embedder of the quantum circuit. We measured observables $M_a$, that is, $M_a=\mathrm{Tr}[O_a\rho(\mathbf{x})]$, $O_a=(\sigma_a +1)/2$, $a \in \{x, z\}$, where $\sigma_x$ and $\sigma_z$ are Pauli operators.
We then prepared a dataset for causal discovery using algebraic operations on the measured values.
Consequently, the data distribution is in general far from a typical probability distribution such as a Gaussian distribution.
This setting aims to highlight that under such data generation processes, the quantum kernels can typically be superior to classical kernels in accurately reproducing the underlying causal structures.
Because the qPC or PC algorithm yields CPDAGs, we evaluate the accuracy by considering Markov equivalence; in this case, the fork and chain should not be distinguished. %

\begin{figure*}[t!]
    \centering
   \includegraphics[width=0.7\textwidth]{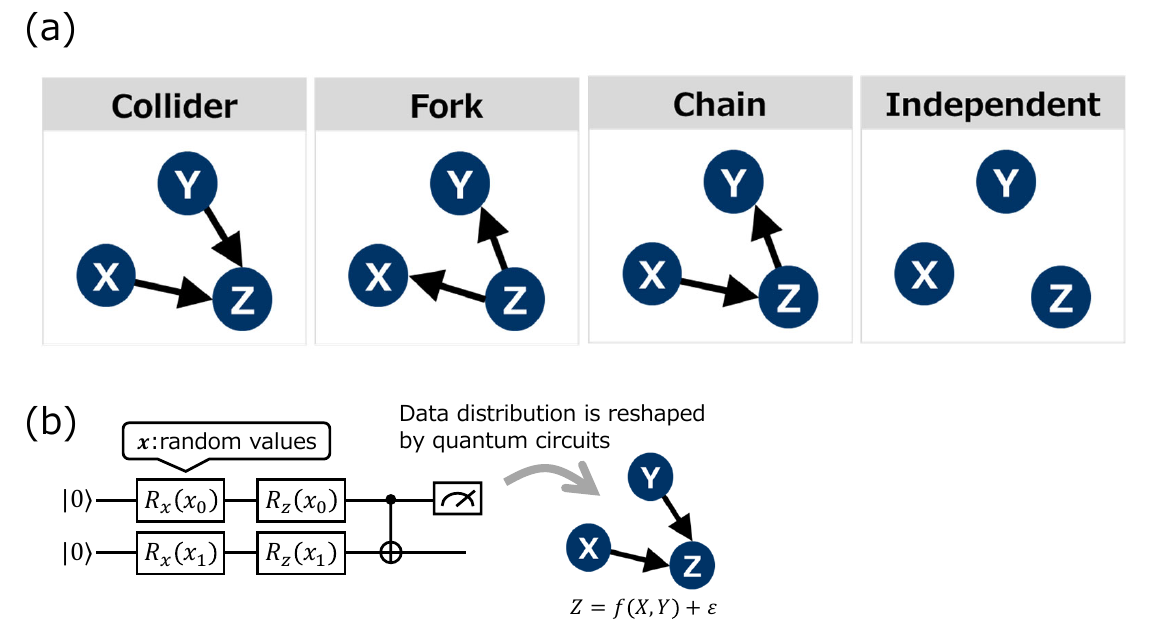}
   \includegraphics[width=0.96\textwidth]{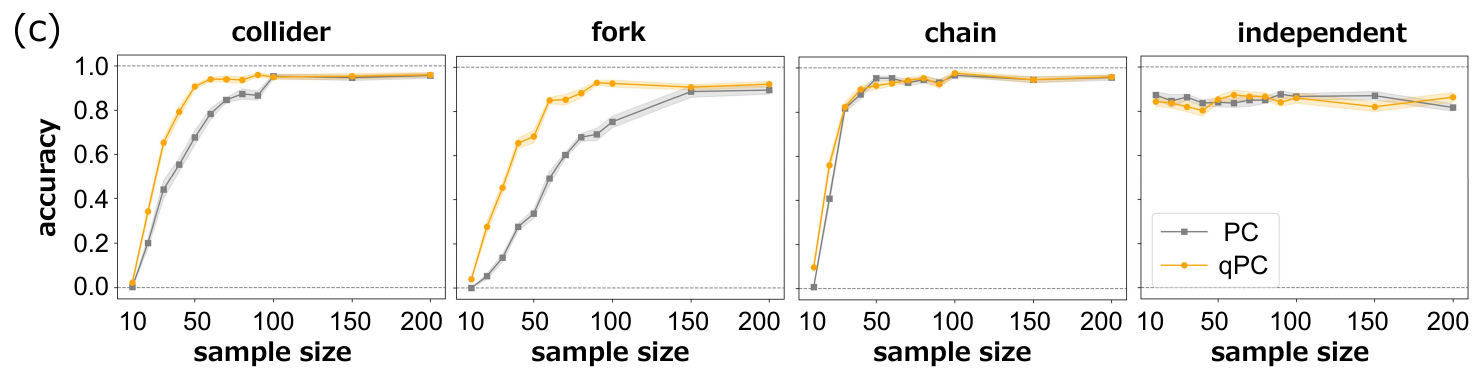}
   \caption{Characteristic performance of the qPC algorithm. 
   (a) Basic causal graphs under three variables with their corresponding dependent and independent relations.
   (b) Data generation with quantum models. The source variables were drawn from quantum circuits with random variable inputs, and the other variables were determined by a causal structure.
   (c) Accuracy of the PC and qPC algorithms for the four causal patterns with different sample sizes. The shaded regions represent the standard errors from 10 different simulations.}
    \label{figure:solo_sumo}
\end{figure*}

Comparisons of the performances of the classical PC and qPC algorithms for causal junctions are shown in Fig.~\ref{figure:solo_sumo} (c).
For chain or independent structures, we observe no significant differences between the classical and quantum methods.
However, for the collider or fork, the quantum kernel outperformed the classical kernel for small sample sizes.
The results of the performance comparison may be questionable since the fork and chain are Markov equivalent. 
However, because the random variable $Z$ constructed from the quantum circuit occupies different positions in the fork and chain, the difficulty of the independence and conditional independence tests in the PC algorithm varies between the fork and chain cases.
In the chain case the random variables are added and mixed with the external noises, while the random variables are not contaminated in the fork case.
The superior performance of the qPC algorithm may have resulted from the inductive bias of the models. 
The data generation process is based on the observation of quantum circuits, which can be related to the quantum kernels used.
In the following sections, we investigate more general cases using datasets unrelated to quantum models.

\subsection{Causal discovery with optimized quantum circuits}
To evaluate the performance of the qPC algorithm using our optimization method, we conducted an experiment in which the data were drawn from a classical setting with the same three fundamental causal graphs as those in Fig.~\ref{figure:solo_sumo} (a).
Figure~\ref{figure:qpc_optimization} (a) shows the typical behaviors of the KTA and the scaling parameter during the optimization process, and the difference in statistics between the default and optimized kernels is shown in Fig.~\ref{figure:qpc_optimization} (b).
Through optimization, the KTA was minimized for the independent data, and correspondingly, the scaling parameter approached the optimal value, as shown in Fig.~\ref{figure:qpc_optimization} (a).
A comparison of the gamma distributions defined in Eq.~(\ref{eq:gamma_dis}), which are the approximation of the distribution of Eq.~(\ref{eq:asymp}), induced by the default and optimized and quantum kernels, is shown in Fig.~\ref{figure:qpc_optimization} (b).
This indicates that the false-positive (FP) probability was substantially suppressed after optimization.
Figure~\ref{figure:qpc_optimization} (c) shows the accuracy over different sample sizes for three cases: the PC with Gaussian kernels of heuristic width choice and the qPC algorithms with quantum kernels of default and optimized scaling parameters.
The qPC algorithm with the default scaling parameters collapses into the collider structure.
However, the optimization of the scaling parameters drastically improved its performance.
The qPC algorithm with optimized parameters performed better than the PC algorithm in the small-size regime.
Figure.~\ref{figure:qpc_optimization} (d) shows the ROC curves for three causal patterns with a sample size $50$.
This suggests that the qPC algorithm, with optimized scaling parameters, can achieve the best performance when the level of significance is set appropriately.
These results indicate that reducing the false-positive (FP) risk yields quantum kernels that surpass classical kernels, even for classical datasets with small sample sizes.

\begin{center}\
\begin{figure*}[t!]
    \centering
    \includegraphics[width=0.90\textwidth]{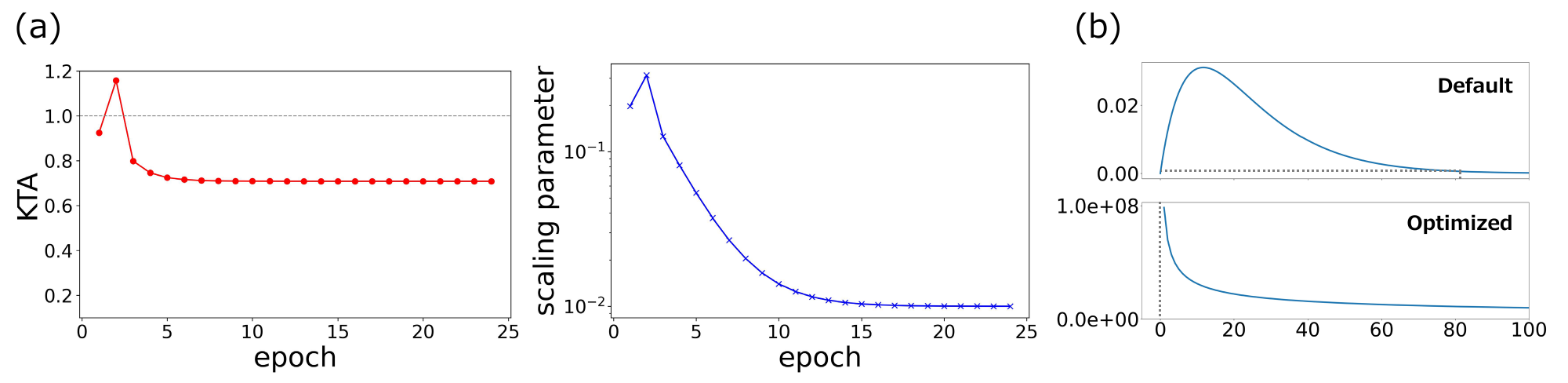}
    \includegraphics[width=0.90\textwidth]{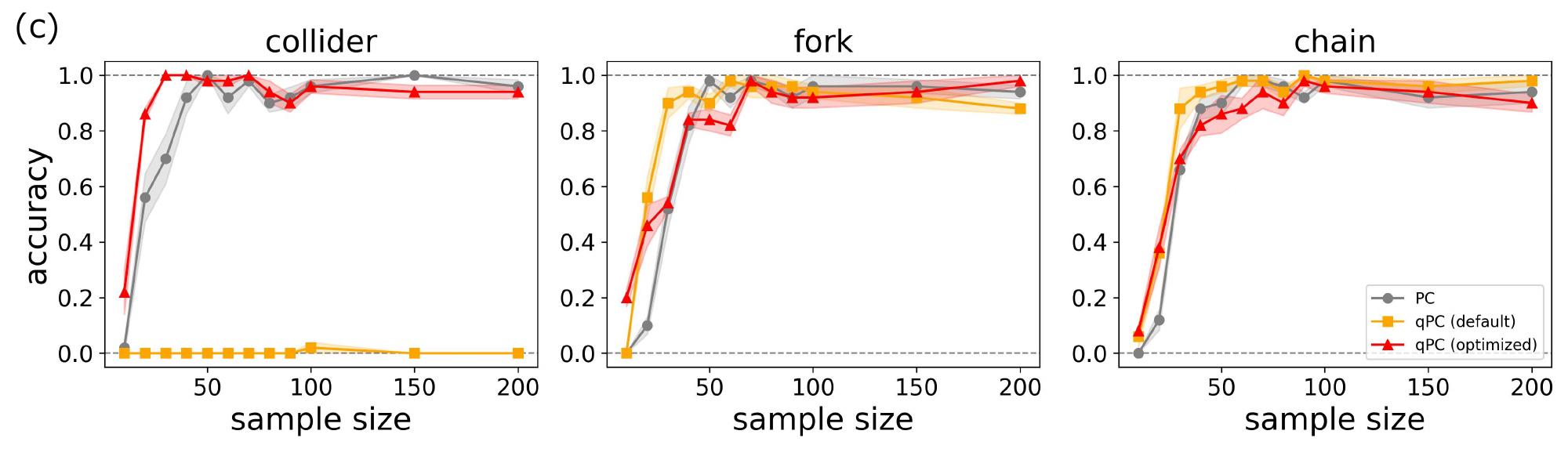}
    \includegraphics[width=0.90\textwidth]{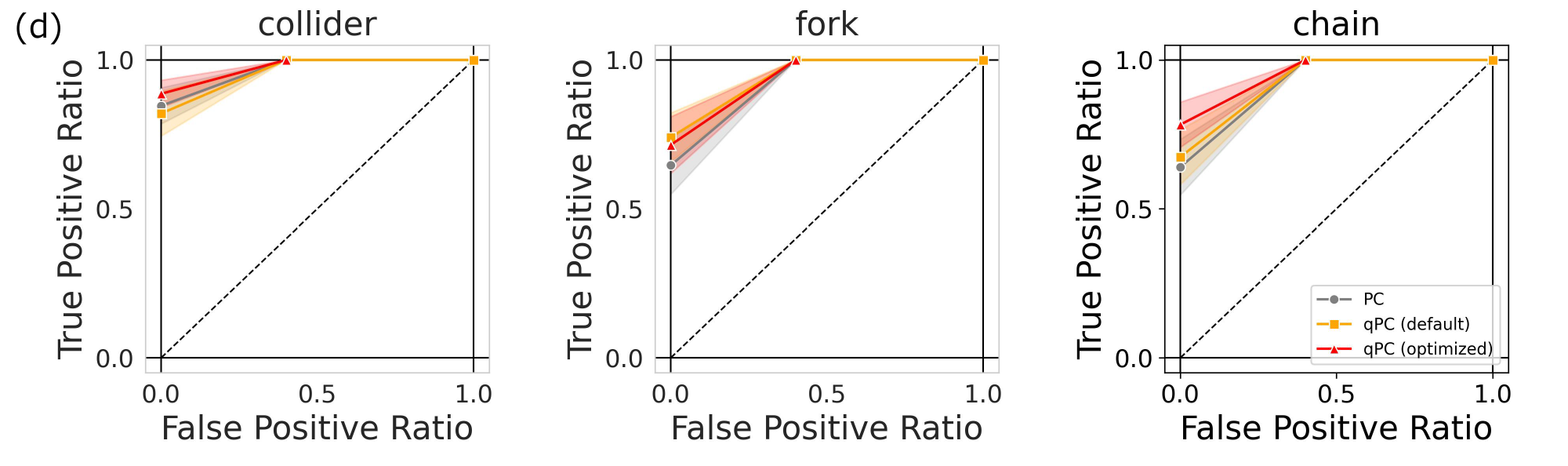}
    \caption{Optimization of the hyperparameters in quantum circuits in the qPC algorithm.
    (a) Changes of the KTA and scaling parameter during optimization. (b) The gamma distribution before and after the optimization process. The endpoint of the dashed box indicates the significance level ($\alpha = 0.05$), corresponding to the tail of the distribution. For (a) and (b), a typical example was chosen from the simulation in (c).
    (c) Accuracy of the PC and qPC with default and optimized hyperparameters with different sample sizes for the three junction patterns.
    (d) ROC curves obtained by the three methods for the junction patterns with $50$ samples.
    The shaded regions represent the standard errors from 10 different simulations.
    In the indendent cases, the three methods showed similar performance, and they are not shown here.}
    \label{figure:qpc_optimization}
\end{figure*}
\end{center}

\subsection{Application of the qPC algorithm to real-world data}\label{sec:real_data}
Here, we demonstrate the application of the qPC algorithm and our optimization method to real-world data.
We used the datasets on the Boston housing price (\cite{harrison1978hedonic}), heart disease (\cite{ahmad2017survival}), and the expression levels of proteins in human immune system cells (\cite{sachs2005causal}).
In the optimization, we sought suitable scaling parameters by minimizing the KTA for the independent distributions obtained by shuffling the original data.

The results of applying the classical PC and qPC algorithms to the Boston housing data are presented in Fig. \ref{figure:boston}.
Panel (a) displays the marginal distributions for the selected variables, most of which appear to deviate significantly from Gaussian or other conventional distributions.
Using the classical PC with KCIT for the full sample data ($N=394$), we obtained the CPDAG shown in Fig. \ref{figure:boston} (b), which captures reasonable causal relations among the variables.
However, the small sample size obscures the causal relations between them, and the PC algorithm failed to reconstruct the CPDAG under the same conditions, such as the level of significance, as shown in Fig. \ref{figure:boston} (c).
The qPC algorithm with optimized scaling parameters remains capable of providing a more comprehensive estimate of causality, as shown in Fig. \ref{figure:boston} (d), where it detects the potential causes of the price, denoted as the MEDV node.
The closeness between the results of the PC with full samples and those of the qPC with a small part of the whole sample set is consistent with our artificial data experiment. 

\begin{figure*}[t]
\includegraphics[width=0.95\textwidth]{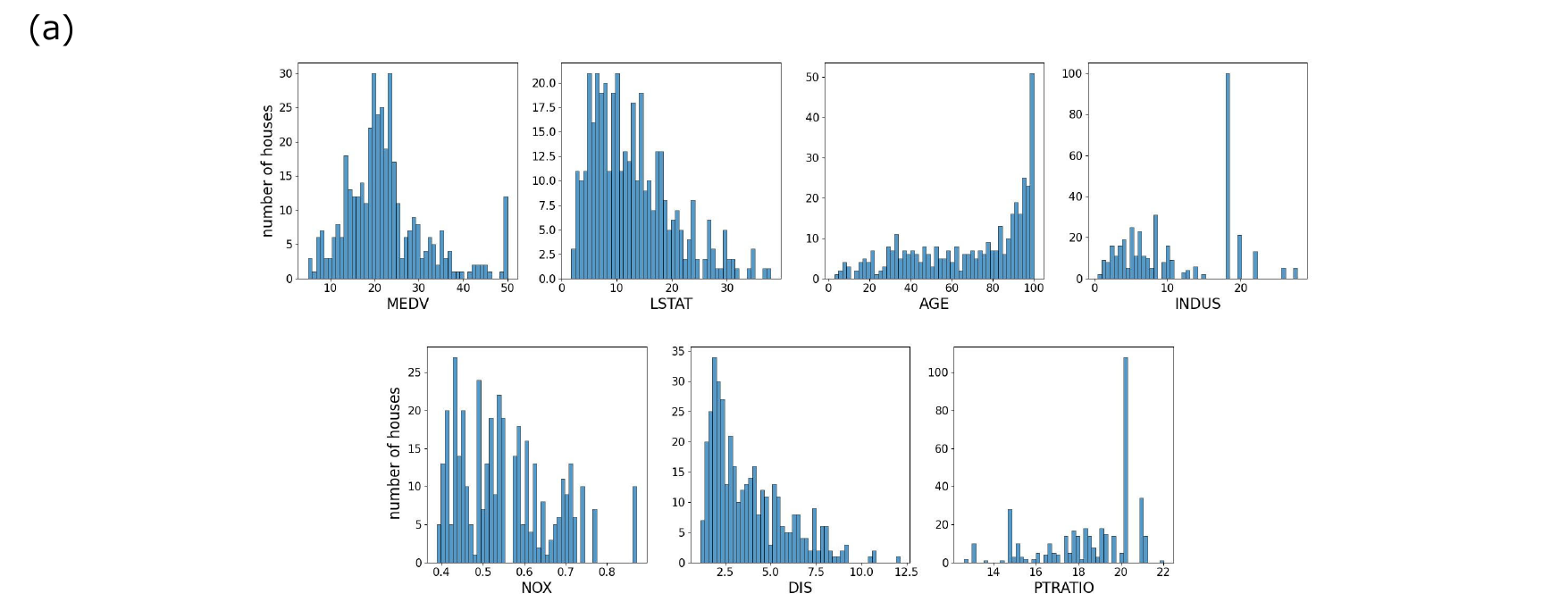}
\includegraphics[width=0.95\textwidth]{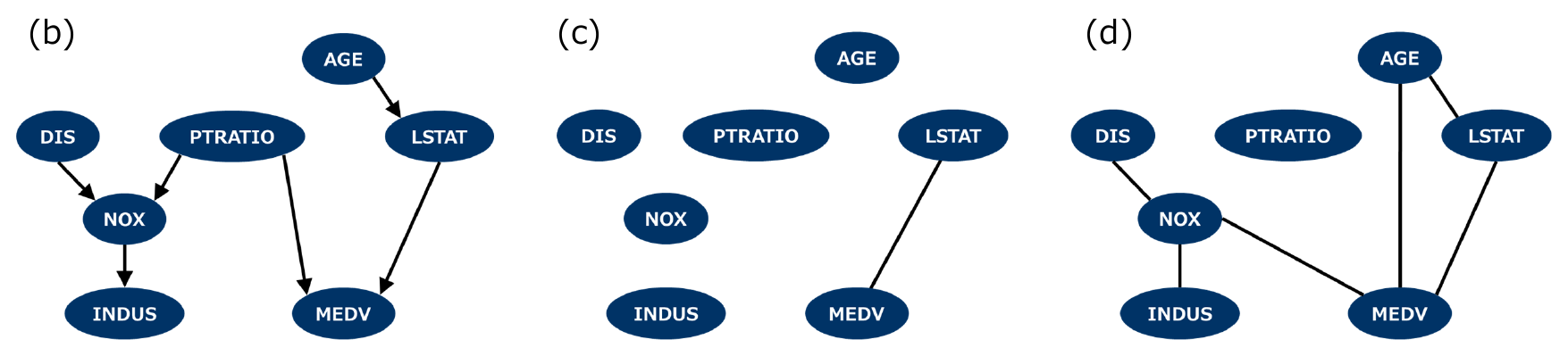}
  \caption{Application to data on housing prices in Boston. 
  (a) Marginal distributions for the variables. 
  (b) CPDAG obtained from the PC algorithm using the Gaussian kernel. The algorithm was executed for the full samples with $N = 394$.
  (c) CPDAG from the PC with a small part of the dataset with $N=50$.
  (d) CPDAG from the qPC using a quantum kernel with the same data as in (c).
  For all cases, the levels of significance were set as $\alpha=0.01.$} \label{figure:boston}
\end{figure*}

We also applied the qPC algorithm to clinical data in which the survival events of heart disease patients and 12 factors were recorded (\cite{ahmad2017survival}).
This dataset comprises 299 patient records, and a previous study (\cite{chicco2020machine}) demonstrated that serum creatinine and ejection fraction are key factors in predicting survival events. These two factors are found to be sufficiently effective in predicting death events in patients with heart failure.
For the full sample set, the classical PC method detected the causal relations between the death event and these two key factors in Fig. \ref{figure:heart} (a).
We showed that for the small subset of the entire datasets$(N=100)$ the qPC with the optimized hyperparameter succeeded in detecting these relations. In contrast, the PC and the qPC with the default hyperparameter did not, as shown in Fig. \ref{figure:heart} (b-d).
In Fig. \ref{figure:heart} (e), we show the performance of the three methods across the sample sizes.
The qPC algorithm with the optimized scaling parameter provided the most accurate description of the causal relations found in the previous study (\cite{chicco2020machine}).
We note that while the qPC algorithm yielded better results for the data on heart disease and housing prices, the performance may depend on the specific data (See Appendix~\ref{app:sachs}).

\begin{figure*}[t!]
    \centering
    \subfigure{\includegraphics[keepaspectratio, scale=0.55]{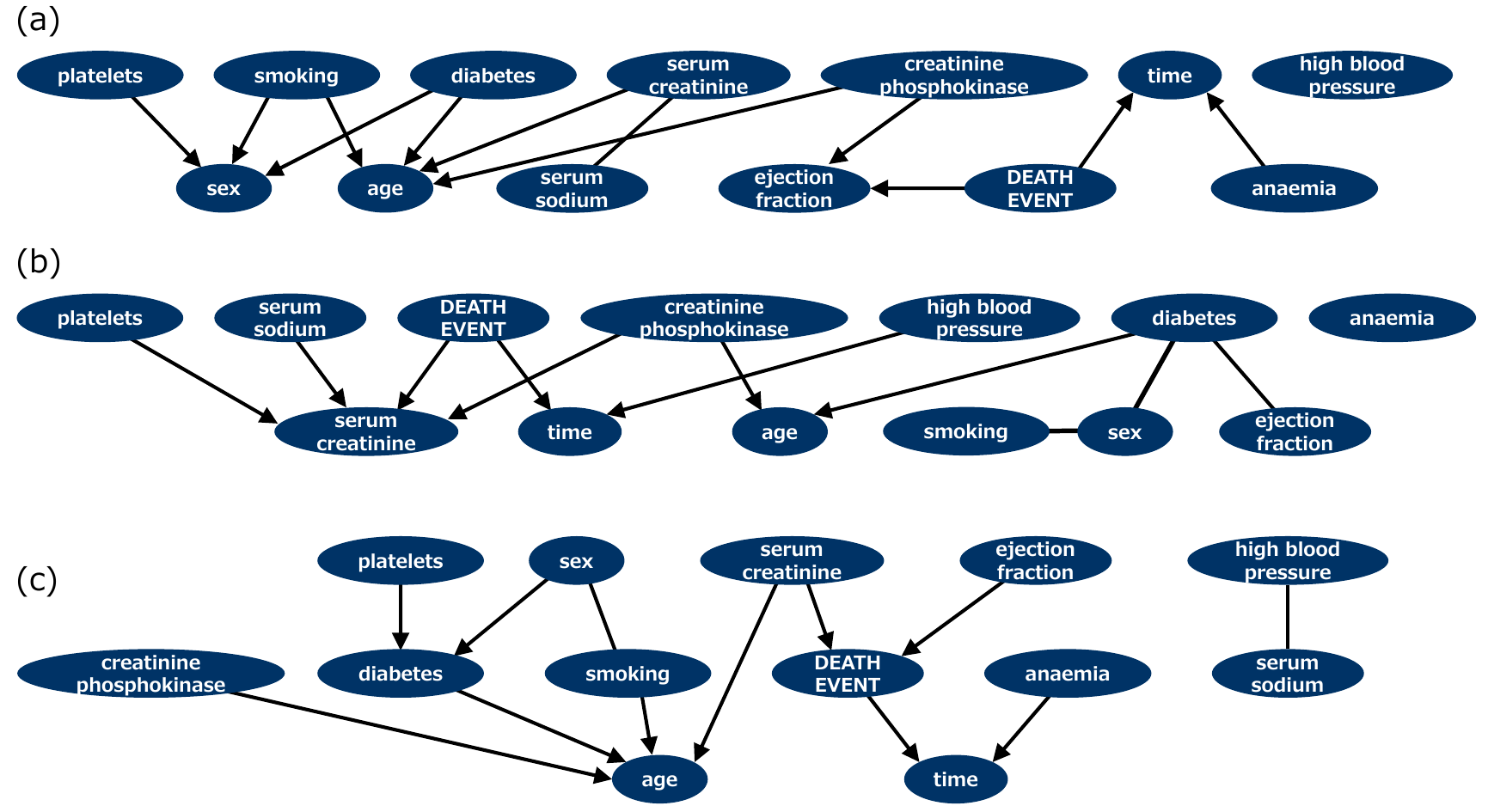}}
    \subfigure{\includegraphics[keepaspectratio, scale=0.32]{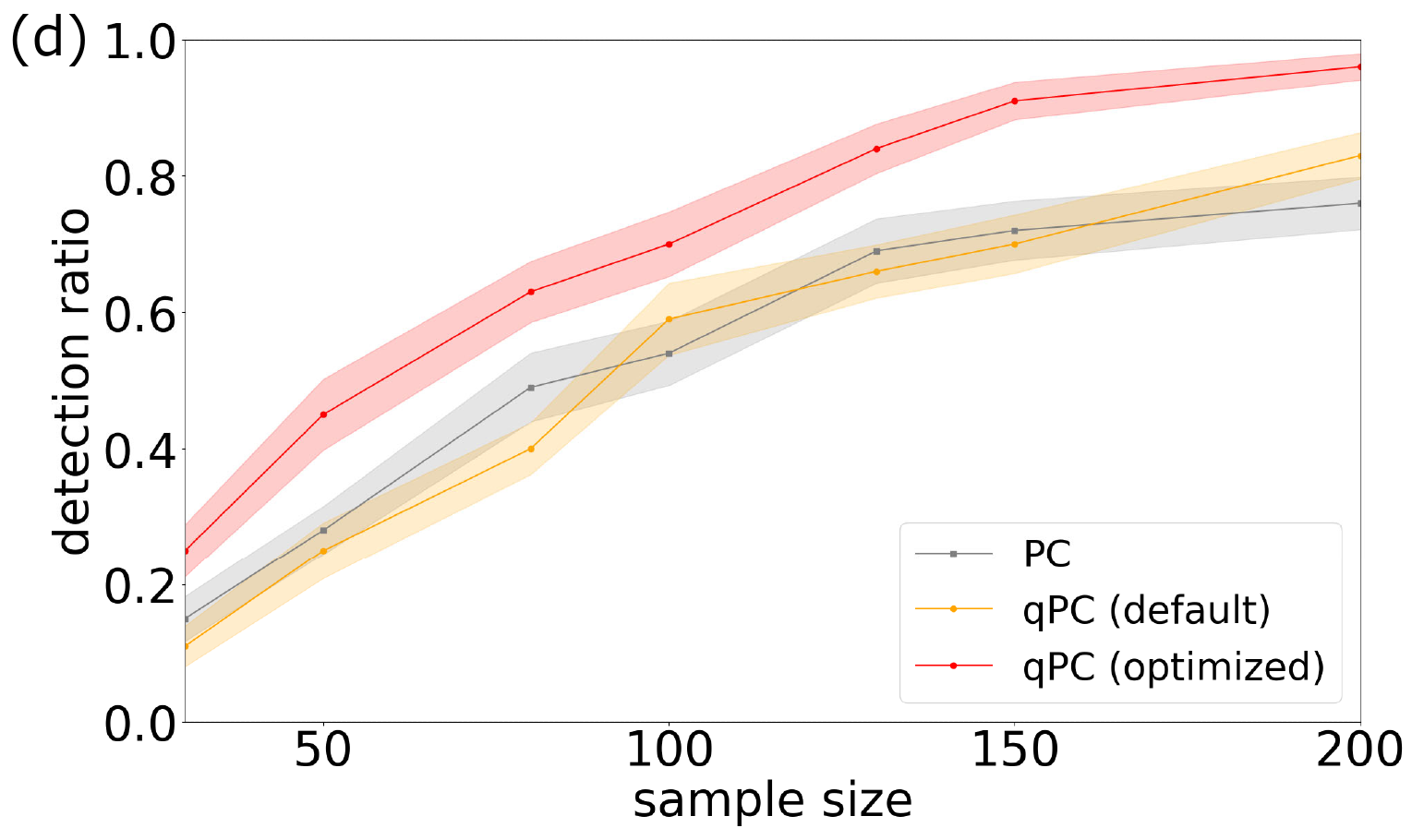}}
    \caption{Application to clinical data on heart disease.
    (a-c) Examples of CPDAGs obtained from different algorithms for the same data.
    (a) PC algorithm using the Gaussian kernel. The algorithm was executed 
    (b) CPDAG obtained from the qPC using a quantum kernel with the default scaling parameter.
    (c) CPDAG obtained from the qPC using a quantum kernel with the scaling parameter optimized via KTA minimization.
    (d) Detection ratios on the links between the death event and the two key factors of serum creatinine and ejection fraction. The shades represent the standard errors over 50 trials.
    For all cases, the levels of significance were set as $\alpha=0.01.$
  }
    \label{figure:heart}
\end{figure*}

\subsection{Experimental details}

Experimental results were generated using the Python package causal-learn (\cite{zheng2024causal}) embedded with our proposed kernel.
We built our quantum models based on the package emulating quantum models with Qiskit (\cite{javadi-abhariQuantumComputingQiskit2024}) and Qulacs (\cite{suzukiQulacsFastVersatile2021}).
In the classical method, we used the KCIT with the heuristic choice of the Gaussian kernel width already implemented in causal-learn, which is one of the methods with the best performance in classical kernels.

In Section \ref{sec:results_solosumo}, our simulations were run with noise ratios $0.05$ for the following relations, where the source variables were drawn from the Gaussian distributions.
In detail, we used the relations of the collider, $z=z_1, x=(z+y)/2, y =x_1^2$, the chain, $z=(z_1+x_1)/2, x=y^2, y=0.5z$, and the fork $z=0.5x, x=(z_1+x_1)/2, y=x^2,$ where $x_1$ and $z_1$ were drawn independently.
To estimate accuracy, we run 30 iterations for each simulation.
The scaling parameters of the quantum models were fixed to $1.0$.
The significance level was set to $\alpha=0.05$.

In Section \ref{sec:results_optimization}, we run our simulation for linear relations with Gaussian variables, unless otherwise described.
For optimization, we created the independent data by shuffling the original data and applied the optimizer to decrease the KTA value of the shuffled data.
We changed the single scaling parameter and searched for its optimal value within the range $[0.01,0.5]$ starting from an initial value of $0.1$.
All data were standardized before applying the causal discovery methods.
In the default quantum models, we used the scaling parameters equivalent to $1$.
In the ROC curves, we changed the level of significance in the set $\{0.999999,\allowbreak\ 0.9,\allowbreak\ 0.75,\allowbreak\ 0.5,\allowbreak\ 0.25,\allowbreak\ 0.2,\allowbreak\ 0.1,\allowbreak\ 0.05,\allowbreak\ 0.01,\allowbreak\ 0.001,\ 0.0001,\ 0.00001\}$. 
The ROC curves require the calculation of the true-positive ratio (TPR) and false-positive ratio (FPR).
We focused on the skeletons of the CPDAGs, considering only the existence or absence of edges between the variables to evaluate the TPR and FPR.
If an edge exists between the two variables, it is judged positive; otherwise, it is judged negative. 
If the estimate and ground-truth match, it is called a true-positive (TP) if an edge is present, and a true negative (TN) if no edge is present. 
Conversely, if the estimate implies that an edge is present and the ground truth does not have an edge, it is called an FP. 
If no edge is inferred in the estimate and an edge is present in the ground truth, it is called an FN.
Using the scores for TP, TN, FP, and FN, TPR and FPR are calculated as $\mathrm{TPR}=\mathrm{TP}/(\mathrm{TP}+\mathrm{FN})$ and $\mathrm{FPR}=\mathrm{FP}/(\mathrm{FP}+\mathrm{TN})$, respectively.

In Section \ref{sec:real_data}, we employed the classical and quantum kernels, which are identical to those used in the previous sections.
For Boston housing data, we used the data source (\cite{bostonkaggleurl}). 
The dataset used for heart disease data can be found in (\cite{heartdeseaseurl}).

\section{Discussion}
We proposed the qPC algorithm for causal discovery by leveraging quantum circuits that generate the corresponding RKHS.
Our simulations demonstrated that the qPC algorithm can surpass the classical method in reconstructing the underlying causal relations, particularly with a small number of samples.
Furthermore, since there is no existing method for determining the hyperparameters of quantum kernels, we propose a method for adaptively choosing quantum kernels for the data.
In the proposed method for kernel choice, we employed the KTA to select quantum kernels suitable for causal discovery, thereby reducing the false-positive (FP) risk for independent cases.
We numerically demonstrated that the optimization method can improve the inference results for both synthetic and real data.
Our experimental results indicate that even for small sizes, quantum kernels can facilitate accurate causal discovery.
This finding suggests that quantum circuits can improve the performance of existing causal discovery methods and expand their applicability to real-world problems.

Although our experiments on artificial and real data suggest the superiority of the qPC algorithm for causal discovery with small datasets compared to the classical PC algorithm, further discussion is needed to unveil the principle behind this phenomenon.
For small sample datasets, we cannot apply the asymptotic theory of the test statistics shown in the KCIT, making it difficult to expect the independence test to perform as theoretically predicted. 
For the KCIT to work effectively for independence tests, data-driven kernel choice may be beneficial; optimization via KTA could enhance the performance of the hypothesis test. 
On the other hand, because such an improvement should be in principle achievable with any kernel, it is reasonable to speculate that the success of the quantum kernel with the dataset used is owing to its inductive bias in quantum models~(\cite{kubler2021inductive}).
Specifically, we observed that optimized quantum kernels tend to exhibit exponentially fast convergence in eigenvalues, which is generally not the case in na\"ive quantum kernels.
We speculate that this property supports effective low-dimensional expression for data and appropriately conducts independence tests.
Although we demonstrated that the qPC algorithm exhibits high accuracy for data generated from quantum circuits, even with default hyperparameters, it fails to capture causal relations from classical data without adjusting the hyperparameters.
Optimization significantly enhances the capacity of the qPC algorithm, making it superior to classical heuristics.
Investigating the properties of quantum kernels, such as their eigenvalues, could provide insights into the underlying mechanisms.
Moreover, the change in the properties of the RKHS associated with the quantum models through optimization and its effect on the independence tests could be studied.

The proposed optimization method based on the KTA increases the applicability of quantum methods.
Our result, shown in Fig. \ref{figure:qpc_optimization}, connects the quantum method with realistic data.
Remarkably, the optimal values of the scaling parameters obtained in our cases are highly compatible with previous results in a supervised learning setting (\cite{shaydulin2022importance}).
This implies that there are parameter regions in which the computational capacity of the quantum kernels is maximized.
Our results could also be used to develop a procedure for heuristic parameter choice in quantum kernels, similar to the one used for Gaussian kernels.
While we chose kernels by minimizing the KTA to decrease the false-positive (FP) probability in this study, other strategies for choosing kernels in independence tests or causal discovery exist.
A study designed kernels for independence tests to maximize test power (\cite{xu2024learning,pogodin2024practical,ren2024learning}).
The main difference is that our method selects kernels to minimize the probability of Type-I errors, whereas their methods aim to reduce Type-II errors.
Another study minimized mutual information (\cite{wang2024optimal}), assuming ridge regression.
In their method, the mutual information is calculated for the obtained causal structures.

Finally, we describe the promising extensions of this study.
First, for simplicity, we assume that no hidden variables affect the causality of the visible variables.
Such confounding factors may change the inferred causal structures. An extended version that incorporates their existence, the FCI algorithm, has been developed (\cite{spirtes2013causal}).
Our algorithm can be used for independence tests within the framework of the FCI algorithm.
In addition, while we focus on static situations in which data are drawn from static distributions, causal discovery has been applied to real-world problems associated with dynamic systems.
Our approach with quantum kernels can be utilized to analyze time-series data with straightforward modifications following the PCMCI algorithm (\cite{runge2019detecting}), which expands the applicability of the qPC algorithms to real-world problems such as meteorology or financial engineering.
In addition, it is possible to develop a more elaborate kernel choice, such as the multiple kernel method (\cite{vedaie2020quantum}), where a combination of multiple kernels is employed, and the optimal solution is obtained via convex optimization. 
These developments will enhance the applicability of the qPC algorithm to various real-world applications.

The present work demonstrates that the quantum-enhanced algorithm can enhance the accuracy of the causal discovery method, particularly for small sample sizes.
Our numerical investigation revealed that the quantum method reconstructed the causal fundamental structures more accurately from small datasets than the classical one.
The introduction of KTA optimization enables us to evaluate optimal quantum kernels without relying on the underlying causal relations.
While the KTA metric provides insights into the types of kernels that yield accurate inference by reducing the false-positive (FP) ratio for independent data, it is not fully understood how the quantum nature elevates the performance of classical methods.
Furthermore, we primarily analyzed the linear cases of causal relations in numerical demonstrations as the initial assessment of the quantum algorithm.
Future work on data with more complicated causal relations or various distributions could offer fundamental insights for practical applications.

\mbox{}
\\
{\bf Acknowledgements:} 
The authors are grateful to Dr. Hiroki Tetsukawa for fruitful discussion.

\mbox{}
\\
{\bf Conflict of interest:} 
The authors declare no competing interests. 

\mbox{}
\\
{\bf Author contribution:}
Y. Maeda and H. Tezuka contributed to the study conception and design.
Y. Terada and Y. Tanaka contributed to manuscript preparation. 
Y. Terada, K. Arai, Y. Tanaka, Y. Maeda, and H. Tezuka commented on and revised the previous versions of the manuscript.
K. Arai, Y. Terada, and H. Tezuka developed the base computation system and conducted experiments to collect and analyze data. 
Y. Terada, Y. Tanaka, K. Arai, and H. Tezuka created all images and drawings.
All the authors have read and approved the final manuscript.

\mbox{}
\\
{\bf Data availability statement:}
The datasets generated and analyzed during the current study are available from the corresponding author upon reasonable request.

\bibliographystyle{spbasic}      %
\bibliography{main}

\newpage
\appendix

\section*{Appendices}
\section{PC algorithm}\label{app:pc}

Here, we summarize the PC algorithm (\cite{spirtes1991algorithm,spirtes2001causation}) and highlight our contribution by emphasizing the difference between the qPC and conventional PC algorithms.
Historically, the PC algorithm (\cite{spirtes1991algorithm}) was introduced as a computationally efficient version of the Spirtes--Glymour--Scheines algorithm and has been widely used due to its efficiency and effectiveness, as it can perform several tests that grow exponentially with the number of variables.
The PC algorithm includes a (conditional) independence test and orientation of the edges to provide the CPDAGs from observed data under the assumptions of causal faithfulness and causal sufficiency.
A CPDAG with directed and undirected edges describes an equivalence class of DAGs and a set of DAGs with the same skeleton and collider structures.
This equivalence class is referred to as a Markov equivalence class.
The causal faithfulness condition states that if two variables are statistically independent, there should be no direct causal path between them in the causal model.
Causal sufficiency assumes that there are no unobserved variables.
The PC algorithm assumes acyclicity in the causal graphs.
We also assume that the observed data are collected independently and are identically distributed.
In contrast to causal model-based algorithms and gradient-based algorithms using statistical models, such as LiNGAM (\cite{10.5555/1248547.1248619}) and NOTEARS (\cite{zheng2018dags}), the PC algorithm does not require any specific functional assumptions on causal relations.
Additionally, the PC algorithm employs statistical tests but does not assume their specific types. Thus, it is applicable to discrete and continuous variables, with suitable tests.
We describe the PC algorithm procedure for obtaining CPDAGs below.

The PC algorithm begins with a complete undirected graph and proceeds through three steps to obtain the CPDAG.
As the first part of the PC algorithm, the skeleton, i.e., the undirected graph corresponding to the CPDAG, was inferred through statistical tests.
In this step, we select two variables from the set of all variables, $X$ and $Y$.
Thereafter, for $X$ and $Y$, we perform an independence test to investigate whether $X \independent Y$.
If the two variables are independent, we remove the edge between them.
For $X$ and $Y$ with a still existing edge and another variable $Z_1$, we perform the conditional independent test to investigate whether $X \independent Y | Z_1$.
For $X$ and $Y$ with a still existing edge and a set of other variables such as $Z_1$ and $Z_2$, we perform the conditional independence test to investigate whether $X \independent Y | Z_1, Z_2$.
The above process continues until the number of other variables $Z_1,Z_2,\cdots$ equals the total number adjacent to $X$ or $Y$.
This process was performed for each ordered pair of variables.
In the second part, one seeks v-structures and orients them as colliders.
In the obtained skeleton graph, if there are edges between $X$ and $Z$ as well as $Y$ and $Z$ but no edge exists between $X$ and $Y$, such as $X-Z-Y$, we investigate whether $X \notindependent Y | Z$.
If this holds true, we call this triplet a v-structure and orient it as a collider, where $X\to Z\gets Y$.
Finally, the remaining parts of the graph were oriented using orientation propagation.
If we find structures such as $X\to Z-Y$, we orient them as $X\to Z\to Y$, given that a v-structure $X\to Z\gets Y$ contradicts $X \independent Y | Z$, as confirmed in the first part.
If we find a structure $X-Y$ with a directed path from $X$ to $Y$, we orient it as $X\to Y.$

Although the PC algorithm is generally applicable, it has inherent limitations associated with its underlying assumptions.
One of the most significant limitations of this study is the presence of confounding factors.
In most real-world problems, the effects of hidden variables cannot be avoided, which breaks the assumptions of the PC algorithm and can thus produce unreliable results.
The FCI algorithm (\cite{10.5555/2074158.2074215}) is a variant of the PC algorithm, and applies to cases with confounders.
In contrast to the PC algorithm, the FCI algorithm determines the directions of arrows when they can be an arrow or a tail.
Consequently, the FCI algorithm yields partial ancestral graphs, which may include not only directed and undirected edges but also bidirected edges representing latent confounders.
Although the FCI algorithm incurs a computational cost, it can be applied in broader situations.
Another problem can arise from assuming static data properties.
The real data we analyze often has temporal structures, which we refer to as time-series data.
In such cases, the PC algorithm can be applied by expanding the causal graphs in the temporal direction.
In both cases, the qPC algorithm can be applied with modifications to the PC algorithm.

\begin{algorithm}[h]
\caption{PC algorithm} 
\begin{algorithmic}[1]
\Procedure{PC Algorithm}{$Data, \alpha, Param $}

\State $V \gets$ set of all variables in $Data$
\State $G \gets$ Complete undirected graph on node set $V$
\State $Kernel \gets$ set of all kernel parameters in $Param$ 
\State // 1. Unconditional Independence Test
\ForAll{pairs of variables $X, Y$ in $V$}
    \If{$IndepTest(X, Y) > \alpha$} \Comment{Kernel-based unconditional independence test}
        \State Remove edge $X - Y$ from $G$
        \State $Sepset(X,Y) \gets \emptyset$
    \EndIf
\EndFor
\State $n \gets 1$ \Comment{Conditioning set size}

\State // 2. Conditional Independence Test
\While{$\exists$ adjacent vertices $X, Y$ with $|adj(G,X)\setminus\{Y\}| \geq n$}
    \ForAll{adjacent vertices $X, Y$ in $G$}
        \ForAll{$S \subseteq adj(G,X)\setminus\{Y\}$ with $|S| = n$}
            \If{$IndepTest(X, Y | S) > \alpha$} \Comment{Kernel-based conditional independence test}
                \State Remove edge $X - Y$ from $G$
                \State $Sepset(X,Y) \gets S$
                \State \textbf{break}
            \EndIf
        \EndFor
    \EndFor
    \State $n \gets n + 1$
\EndWhile

\State // 3. Orient the edges in the Graph $G$
\ForAll{subgraph $X - Z - Y$ in $G$, where $X$ and $Y$ are not adjacent}
    \If{$Z \notin Sepset(X,Y)$}
        \State Orient $X - Z - Y$ as $X \to Z \leftarrow Y$.
    \EndIf
\EndFor

\ForAll{subgraph $X \rightarrow Z - Y$ in $G$, where $X$ and $Y$ are not adjacent}
    \State Orient $Z - Y$ as $Z \rightarrow Y$.
\EndFor

\ForAll{subgraph $X - Y$ in $G$ with a directed path from $X$ to $Y$}
    \State Orient $X - Y$ as $X \rightarrow Y$.
\EndFor

\State \Return $G$ \Comment{Partially directed acyclic graph}

\EndProcedure
\end{algorithmic}
\end{algorithm}

\begin{center}\
\begin{figure*}[t!]
    \centering
    \includegraphics[width=0.85\textwidth]{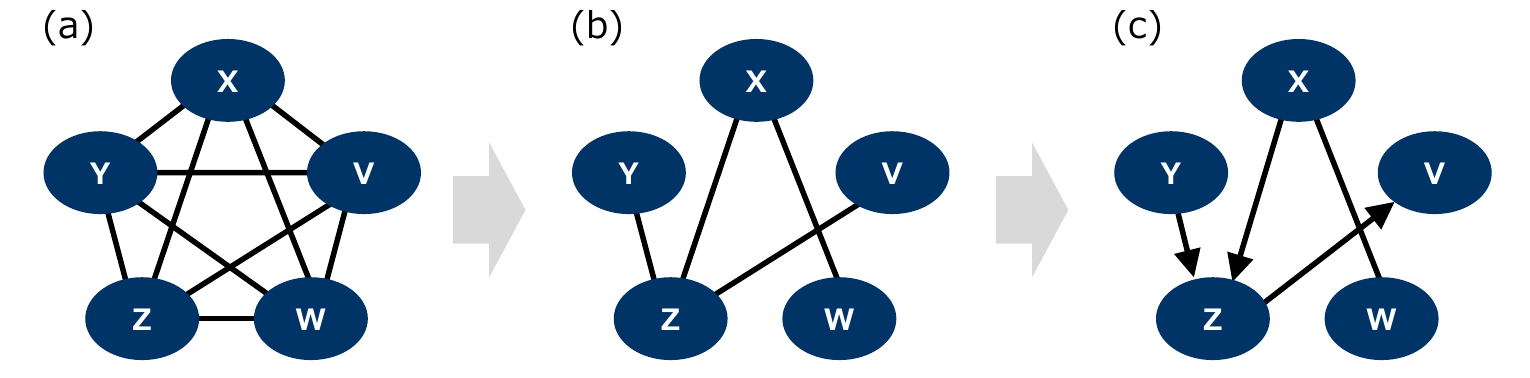}
    \caption{Schematic of the process of the PC algorithm. It begins with the complete graph, as shown in (a). (Conditional) Independence tests are executed to remove edges among them as in (b). Orientation rule gives the arrows their orientations if the conditions are satisfied, as in (c).}
    \label{figure:pc}
\end{figure*}
\end{center}\

\begin{figure*}[t!]
    \centering
    \includegraphics[keepaspectratio, scale=0.6]{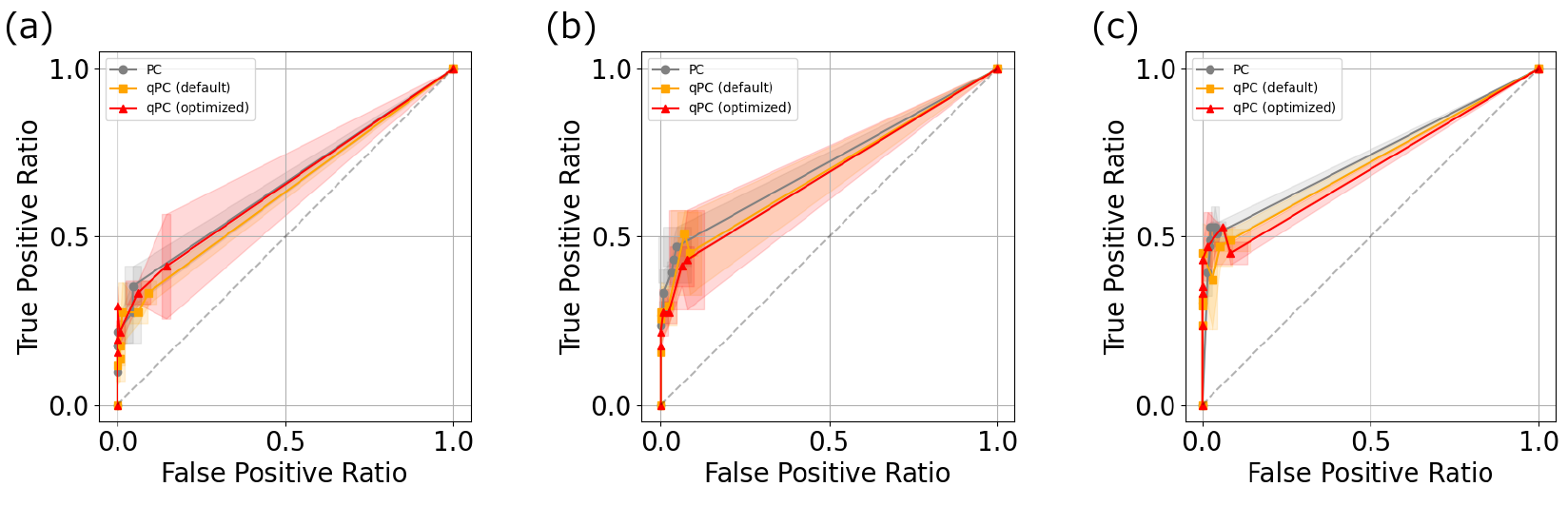}
    \caption{Application to gene expression data with the gold standard network.
    ROC curves for the PC and qPC algorithms for different sample sizes.
    (a) $N=30.$ (b) $N=80.$ (c) $N=400.$
  }
    \label{figure:sachs}
\end{figure*}

\section{Review of the kernel-based conditional independence test}\label{app:kcit}

This section provides a brief review of the KCIT (\cite{10.5555/3020548.3020641, zhang2012kernelbasedconditionalindependencetest}).
Let us begin with given continuous random variables $X, Y$, and $Z$ with domains $\mathcal{X}, \mathcal{Y}$, and $\mathcal{Z}$, respectively. 
The probability law for $X$ is denoted by $P_X$.
We introduce a measurable, positive definite kernel $k_{\mathcal{X}}$ on $\mathcal{X}$ and denote the corresponding RKHS as $\mathcal{H}_{\mathcal{X}}$.
The space of the square integrable functions of $X$ is denoted by $L^2_X$.
$\mathbf{K}_X$ is then the kernel matrix of the {\it i.i.d.} sample $\mathbf{x} = \{ x_1, ..., x_n \}$ of $X$, and $\widetilde{\mathbf{K}}_X = \mathbf{H} \mathbf{K}_{\mathbf{x}} \mathbf{H}$ is the centralized kernel, where $\mathbf{H} := \mathbf{I} - \frac{1}{n} \mathbf{1} \mathbf{1}^T$ with $\mathbf{I}$ and $\mathbf{1}$ being the $n \times n$ identity matrix and the vector of 1's, respectively.
Similarly, we define $P_Y, P_Z, k_{\mathcal{Y}, k_{\mathcal{Z}}}, \mathcal{H}_{\mathcal{Y}}, \mathcal{H}_{\mathcal{Z}},\allowbreak\ L^2_Y,\allowbreak\ L^2_Z,\allowbreak\ \mathbf{K}_Y,\allowbreak\ \mathbf{K}_Z, \widetilde{\mathbf{K}}_Y, \widetilde{\mathbf{K}}_Z$ as well.

The problem here is to perform the test for conditional independence (CI), {\it i.e.}, test the null hypothesis $X \perp \!\!\!\! \perp Y \ | \ Z$, between $X$ and $Y$ given $Z$ from their {\it i.i.d.} samples.
In Refs. (\cite{10.5555/3020548.3020641, zhang2012kernelbasedconditionalindependencetest}), a CI test was developed by defining a simple statistic based on two characterizations of the CI~(\cite{10.5555/2981562.2981624, 10.1093/biomet/67.3.581}) and deriving its asymptotic distribution under the null hypothesis. 

One characterization of the CI is provided in terms of the cross-covariance operator $\Sigma_{XY}$ in the RKHS~(\cite{10.5555/2981562.2981624}).
For random vector $(X, Y)$ on $\mathcal{X} \times \mathcal{Y}$, cross-covariance operator $\Sigma_{XY}$ is defined by the following relation:
\begin{eqnarray}
    \langle f, \Sigma_{XY}f \rangle = \mathbb{E}_{XY} \left[ f(X) g(Y) \right] - \mathbb{E}_X \left[ f(X) \right] \mathbb{E}_Y \left[ g(Y) \right] 
\end{eqnarray}
for all $f \in \mathcal{H}_{\mathcal{X}}$ and $g \in \mathcal{H}_{\mathcal{Y}}$.
\begin{lemma}[Theorem 3 (ii) of Ref. (\cite{10.5555/2981562.2981624})] \label{Fukumizu_2007}
    Denote $\ddot{X} = (X, Z)$ and $k_{\ddot{\mathcal{X}}} = k_{\mathcal{X}} k_{\mathcal{Z}}$.
    Assume that $\mathcal{H}_{\mathcal{X}} \subset L^2_X, \mathcal{H}_{\mathcal{Y}} \subset L^2_Y$, and $\mathcal{H}_{\mathcal{Z}} \subset L^2_Z$.
    Furthermore, assume that $k_{\ddot{\mathcal{X}}} k_{\mathcal{Y}}$ is a characteristic kernel on $(\mathcal{X} \times \mathcal{Z}) \times \mathcal{Y}$ and  $\mathcal{H}_{\mathcal{Z}} + \mathbb{R}$ is dense in $L^2(P_Z)$.
    Then, 
    \begin{eqnarray}
        \Sigma_{\ddot{X}Y | Z} = 0 \Leftrightarrow X \perp \!\!\!\! \perp Y \ | \ Z.
    \end{eqnarray}
\end{lemma}

The other characterization of CI is given by explicitly enforcing the uncorrelatedness of functions in
suitable spaces.
\begin{lemma}[{(\cite{10.1093/biomet/67.3.581})}]\label{Daudin_1980} 
    The following conditions are equivalent to each other:
    \begin{eqnarray}
        X \perp \!\!\!\! \perp Y \ | \ Z \Leftrightarrow \mathbb{E}\left[ f' g' \right] = 0, \forall f' \in \mathcal{E}_{XZ}, \forall g' \in \mathcal{E}'_{XZ}, \nonumber \\
    \end{eqnarray}
    where 
    \begin{eqnarray}
        \mathcal{E}_{XZ} &:=& \left\{ f' \in L^2_{\ddot{X}} \ | \ \mathbb{E}\left[ f' | Z \right] = 0 \right\}, \\
        \mathcal{E}'_{YZ} &:=& \left\{ g' \ | \ g' = g(Y) - \mathbb{E}\left[ g | Z \right], g \in L^2_Y \right\}.
    \end{eqnarray}
\end{lemma}
These functions are constructed from the corresponding $L^2$ spaces. 
For instance, for arbitrary $f \in L^2_{XZ}$, function $f'$ is given by
\begin{eqnarray}
    f'(\ddot{X}) = f(\ddot{X}) - \mathbb{E}\left[ f | Z \right] = f(\ddot{X}) - h_f^{\ast}(Z),
    \label{eq:reg_fun}
\end{eqnarray}
where $h_f^{\ast} \in L^2_Z$ denotes regression function $f(\ddot{X})$ on $Z$. 

Refs.~(\cite{10.5555/3020548.3020641, zhang2012kernelbasedconditionalindependencetest}) established that if functions $f$ and $g$ are restricted to 
spaces $\mathcal{H}_{\ddot{X}}$ and $\mathcal{H}_Y$, respectively, then Lemma~\ref{Daudin_1980} is reduced to Lemma~\ref{Fukumizu_2007}. 
Specifically, they used kernel ridge regression to estimate the regression function $h_f^{\ast}$ in Eq.~\eqref{eq:reg_fun}; that is, 
\begin{eqnarray}
    \hat{h}_f^{\ast} (\mathbf{z}) = \widetilde{\mathbf{K}}_Z (\widetilde{\mathbf{K}}_Z + \epsilon \mathbf{I})^{-1} \cdot f(\ddot{\mathbf{x}}), 
    \label{eq:k_ridge}
\end{eqnarray}
where $\epsilon$ denotes a small positive regularization parameter.
From Eq.~\eqref{eq:k_ridge}, we can construct a centralized kernel
matrix corresponding to function $f'(\ddot{X})$,
\begin{eqnarray}
    \widetilde{\mathbf{K}}_{\ddot{X}|Z} = \mathbf{R}_Z \widetilde{\mathbf{K}}_{\ddot{X}} \mathbf{R}_Z, 
\end{eqnarray}
where $\mathbf{R}_Z = \mathbf{I} - \widetilde{\mathbf{K}}_Z (\widetilde{\mathbf{K}}_Z + \epsilon \mathbf{I})^{-1} = \epsilon (\widetilde{\mathbf{K}}_Z + \epsilon \mathbf{I})^{-1}$.
Similarly, we construct a centralized kernel
matrix $\widetilde{\mathbf{K}}_{Y|Z}$ corresponding to function $g'(Y)$.

Furthermore, to propose the statistic for CI, they provided general results on the asymptotic distributions of specific statistics defined in terms of kernel matrices under the assumption of uncorrelatedness between functions in particular spaces.
Let us consider the eigenvalue decompositions of the centralized kernel matrices of $\widetilde{\mathbf{K}}_{X}$ and $\widetilde{\mathbf{K}}_{Y}$, {\it i.e.}, $\widetilde{\mathbf{K}}_{X} = \mathbf{V}_{X} \mathbf{\Lambda}_{X} \mathbf{V}_{X}^T$ and $\widetilde{\mathbf{K}}_{Y} = \mathbf{V}_{Y} \mathbf{\Lambda}_{Y} \mathbf{V}_{Y}^T$, where $\mathbf{\Lambda}_{X}$ and $\mathbf{\Lambda}_{Y}$ are diagonal matrices containing the non-negative eigenvalues $\lambda_{\mathbf{x}, i}$ and $\lambda_{\mathbf{y}, j}$, respectively.
Furthermore, we define that $\boldsymbol{\psi}_{\mathbf{x}} = \left[ \psi_{\mathbf{x}, 1}(\mathbf{x}), ..., \psi_{\mathbf{x}, n}(\mathbf{x}) \right]\allowbreak\ := \mathbf{V}_{X} \mathbf{\Lambda}_{X}^{1/2}$ and $\boldsymbol{\phi}_{\mathbf{y}} = \left[ \phi_{\mathbf{y}, 1}(\mathbf{y}), ..., \phi_{\mathbf{y}, n}(\mathbf{y}) \right] := \mathbf{V}_{Y} \mathbf{\Lambda}_{Y}^{1/2}$, {\it i.e.}, $\psi_{\mathbf{x}, i}(x_k) = \lambda_{\mathbf{x}, i}^{1/2} V_{\mathbf{x}, ik}$ and $\phi_{\mathbf{y}, j}(y_k) = \lambda_{\mathbf{y}, j}^{1/2} V_{\mathbf{y}, jk}$. 
Then, defining tensor $\mathbf{T}$ and matrix $\mathbf{T}^{\ast}$ by 
\begin{eqnarray}
    T_{ijk} &:=& \frac{1}{\sqrt{n}}\psi_{\mathbf{x}, i}(x_k) \phi_{\mathbf{y}, j}(y_k) \\ 
    &=& \sqrt{\frac{\lambda_{\mathbf{x}, i} \lambda_{\mathbf{y}, j}}{n}} V_{\mathbf{x}, ik} V_{\mathbf{y}, jk}, \\
    T^{\ast}_{ij}(X,Y) &:=& \sqrt{\lambda^{\ast}_{X,i} \lambda^{\ast}_{Y,j}} u_{X,i}(X) u_{Y,j}(Y),
\end{eqnarray}
where $\lambda^{\ast}_{X,i}, \lambda^{\ast}_{Y,j}$ and $u_{X,i}(X) u_{Y,j}(Y)$ are the eigenvalues and eigenfunctions of kernel $k_{\mathcal{X}}$ with regard to the probability measure with the density $p(x)$, respectively, we define matrices $\mathbf{M}$ and $\mathbf{M}^{\ast}$ by
\begin{eqnarray}
    M_{ij, i'j'} &=& \sum_{k=1}^n T_{ijk} T_{i'j'k},  \\
    M^{\ast}_{ij, i'j'} &=& T^{\ast}_{ij}(X, Y) T^{\ast}_{i'j'}(X, Y).
    \label{eq:M}
\end{eqnarray}
Note that $\mathbf{M}$ and $\mathbf{M}^{\ast}$ for the conditional kernels are defined similarly.
The main technical results presented in Ref.~(\cite{10.5555/3020548.3020641, zhang2012kernelbasedconditionalindependencetest}) are as follows:

\begin{theorem}[Theorem 3 of Ref.~(\cite{10.5555/3020548.3020641, zhang2012kernelbasedconditionalindependencetest}]
    Suppose that we are given arbitrary centred kernels $k_{\mathcal{X}}$ and $k_{\mathcal{Y}}$ with discrete eigenvalues and the corresponding RKHS’s $\mathcal{H}_{\mathcal{X}}$ and $\mathcal{H}_{\mathcal{Y}}$ for sets of random variables $X$ and $Y$, respectively. We make the following three statements:
    \begin{itemize}
        \item[1)] Under the condition that $f(X)$ and $g(Y)$ are uncorrelated for all $f \in \mathcal{H}_{\mathcal{X}}$ and $g \in \mathcal{H}_{\mathcal{Y}}$, for any $L$ such that $\lambda^{\ast}_{X,L+1} \neq \lambda^{\ast}_{X,L}$ and $\lambda^{\ast}_{Y,L+1} \neq \lambda^{\ast}_{Y,L}$, we have
        \begin{eqnarray}
            \sum_{i,j = 1}^L M_{ij, ij} \overset{d}{\longrightarrow} \sum_{i,j=1}^L \mathring{\lambda}^{\ast}_{ij} z^2_{ij},\quad {\rm as}\ n \to \infty, 
            \label{eq:con1}
        \end{eqnarray}
        where $z_{ij}$ are {\it i.i.d.} standard Gaussian variables ({\it i.e.}, $z_{ij}^2$ are {\it i.i.d.} $\chi^2_1$-distributed variables), and $\mathring{\lambda}^{\ast}_{ij}$ are the eigenvalues of $\mathbb{E}[\mathbf{M}^{\ast}]$.
        \item[2)] In particular, if $X$ and $Y$ are further independent, we have
        \begin{eqnarray}
            \sum_{i,j = 1}^L M_{ij, ij} \overset{d}{\longrightarrow} \sum_{i,j=1}^L \lambda^{\ast}_{X, i} \lambda^{\ast}_{Y, j} z_{ij}^2,\quad {\rm as}\ n \to \infty, 
            \label{eq:con2}
        \end{eqnarray}
        where $z_{ij}^2$ are {\it i.i.d.} $\chi^2_1$-distributed variables.
        \item[3)] The results of Eqs.~\eqref{eq:con1} and ~\eqref{eq:con2} hold for $L = n \to \infty$.
    \end{itemize}
\end{theorem}

Based on these considerations, the authors in Ref.~(\cite{10.5555/3020548.3020641, zhang2012kernelbasedconditionalindependencetest}) proposed statistics defined by the HSIP for unconditional and conditional independence tests.
\begin{theorem}[Theorem 4 of Ref.~(\cite{10.5555/3020548.3020641, zhang2012kernelbasedconditionalindependencetest})]
    Under the null hypothesis that $X$ and $Y$ are statistically independent, statistic
    \begin{eqnarray}
        T_{UI} := \frac{1}{n} {\rm Tr} \bigl[ \widetilde{\mathbf{K}}_X \widetilde{\mathbf{K}}_Y \bigr]
    \end{eqnarray}
    has the same asymptotic distribution as
    \begin{eqnarray}
        \breve{T}_{UI} := \frac{1}{n^2} \sum_{i,j=1}^n \lambda_{\mathbf{x},i} \lambda_{\mathbf{y},j} z_{ij}^2, \label{eq:asymp}
    \end{eqnarray}
    {\it i.e.}, $T_{UI} \overset{d}{=} \breve{T}_{UI}$ as $n \to \infty$, where $z_{ij}$ are {\it i.i.d.} standard Gaussian variables, $\lambda_{\mathbf{x},i}$ are the eigenvalues of $\widetilde{\mathbf{K}}_X$, and $\lambda_{\mathbf{y},i}$ are the eigenvalues of $\widetilde{\mathbf{K}}_Y$.
\end{theorem}

The statistic for the unconditional independence test closely relates to those based on the Hilbert-Schmidt independence criterion (HSIC)~(\cite{NIPS2007_d5cfead9}). 
The difference between these statistics lies in their distinct asymptotic distributions. 
Eq.~\eqref{eq:asymp} depends on the eigenvalues of $\widetilde{\mathbf{K}}_X$ and $\widetilde{\mathbf{K}}_Y$, whereas the HSIC$_b$ in Eq.~(4) in Ref.~(\cite{NIPS2007_d5cfead9}) depends
on the eigenvalues of an order-four tensor.
The following is the statistic for CI.

\begin{theorem}[Theorem 5 of Ref.~(\cite{10.5555/3020548.3020641, zhang2012kernelbasedconditionalindependencetest}]
    Under the null hypothesis that $X$ and $Y$ are conditionally independent, given $Z$, we obtain the statistic
    \begin{eqnarray}
        T_{CI} := \frac{1}{n} \mathrm{Tr} \bigl[ \widetilde{\mathbf{K}}_{\ddot{\mathbf{X}} | \mathbf{Z}} \widetilde{\mathbf{K}}_{\mathbf{Y} | \mathbf{Z}} \bigr]
    \end{eqnarray}
    has the same asymptotic distribution as
    \begin{eqnarray}
        \breve{T}_{CI} := \frac{1}{n} \sum_{k=1}^{n^2} \lambda_k z_k^2,
    \end{eqnarray}
    where $\lambda_k$ are the eigenvalues of matrix $\mathbf{M}$ in Eq.~\eqref{eq:M}, which is constructed by $\widetilde{\mathbf{K}}_{\ddot{\mathbf{X}} | \mathbf{Z}}$ and $\widetilde{\mathbf{K}}_{\mathbf{Y} | \mathbf{Z}}$,
    and $z_k$ are {\it i.i.d.} standard Gaussian variables.
\end{theorem}

We can construct the unconditional and conditional independence tests by generating approximate
null distribution using the Monte Carlo simulation. 
In practice, we can approximate the null distribution with a gamma distribution whose two parameters are related to the mean
and variance. 
Under the null hypothesis, the distribution of $\breve{T}_{UI}$ can be approximated by the $\Gamma(k, \theta)$ distribution
\begin{eqnarray}
    p(t) = t^{k-1} \frac{e^{-t/\theta}}{\theta^k \Gamma(k)},\label{eq:gamma_dis}
\end{eqnarray}
where $k = \mathbb{E}^2 \bigl[\breve{T}_{UI} \bigr] / \mathbb{V}ar \bigl[\breve{T}_{UI} \bigr]$ and $\theta = \mathbb{V}ar \bigl[\breve{T}_{UI} \bigr] / \mathbb{E} \bigl[\breve{T}_{UI} \bigr]$. In the unconditional case, the two parameters can be defined similarly.
The mean and variance are estimated as follows:

\begin{theorem}[Proposition 5 of Ref.~(\cite{10.5555/3020548.3020641, zhang2012kernelbasedconditionalindependencetest}] \label{thm:gamma_param}
    \begin{itemize}
        \item[1)] Under the null hypothesis that $X$ and $Y$ are independent, on the given sample $\mathcal{D}$, we have that
        \begin{eqnarray}
            \mathbb{E} \bigr[ \breve{T}_{UI} | \mathcal{D} \bigl] &=& \frac{1}{n^2} {\rm Tr} \bigr[ \widetilde{\mathbf{K}}_X \bigl] {\rm Tr} \bigr[ \widetilde{\mathbf{K}}_Y \bigl], \label{eq:expectation} \\
            \mathbb{V}ar \bigr[ \breve{T}_{UI} | \mathcal{D} \bigl] &=& \frac{2}{n^4} {\rm Tr} \bigr[ \widetilde{\mathbf{K}}_X^2 \bigl] {\rm Tr} \bigr[ \widetilde{\mathbf{K}}_Y^2 \bigl]. \label{eq:variance}
        \end{eqnarray}
        \item[2)] Under the null hypothesis that $X$ and $Y$ are conditionally independent given $Z$, we have that
        \begin{eqnarray}
            \mathbb{E} \bigr[ \breve{T}_{CI} | \mathcal{D} \bigl] &=& {\rm Tr} \bigr[ \mathbf{M} \bigl], \\
            \mathbb{V}ar \bigr[ \breve{T}_{CI} | \mathcal{D} \bigl] &=& 2 {\rm Tr} \bigr[ \mathbf{M}^2 \bigl],
        \end{eqnarray}
        where $\mathbf{M}$ is the matrix of Eq.~\eqref{eq:M}, which is constructed by $\widetilde{\mathbf{K}}_{\ddot{\mathbf{X}} | \mathbf{Z}}$ and $\widetilde{\mathbf{K}}_{\mathbf{Y} | \mathbf{Z}}$.
    \end{itemize}
\end{theorem}

\section{Details of quantum circuits}\label{app:qc}
Here, we describe the quantum circuit candidates used in this study.
As described in Sec. \ref{subsec:qkct}, the structure of quantum circuit $U(\mathbf{x})$, called as ``ansatz," is composed of three parts: the initialization $U_\mathrm{init}$, data embedding $U_{\mathrm{emb}}(\mathbf{x})$, and entangling $U_{\mathrm{enc}}$ parts, as shown in Fig.~\ref{fig:qc_whole}.
In addition, the amount of data reuploaded, referred to as the depth $n_{\mathrm{dep}}$, is a significant degree of freedom in quantum circuits.
We compared the performance of the causal discovery problems with various combinations of components.
This lineup is illustrated in (Fig.~\ref{fig:qc_parts}) as follows: $U_{\mathrm{init}} \in \{\mathrm{None}, H, S, T\}$, $U_{\mathrm{emb}}(\mathbf{x}) \in \{RY, RXRZ\}$, $U_{\mathrm{ent}} \in \{CX, CZ, \sqrt{\mathrm{iSWAP}}\}\{\mathrm{ladder, circ, all\_to\_all}\}$, and $n_{\mathrm{dep}} \in \{1,\allowbreak\ 4,\allowbreak\ 16\}$ for junction pattern experiments and $n_{\mathrm{dep}} \in \{5\}$ for real world data experiments.
These candidates were partially selected based on the expressibility reported by (\cite{sim2019expressibility}) and (\cite{haug2021capacity}); however, we did not observe a clear correlation between ansatz expressibility and causal discovery performance.

Finally, we describe the quantum circuit used to generate the dataset in Sec.~\ref{sec:results_solosumo} in Fig. ~\ref{fig:qc_datagen}.
Using this data generator, input vector $\mathbf{x} \in [0, \pi]^2$ is mapped to $[0, 1]^2$ via quantum operation.
We found that analyzing the dataset generated by this procedure is difficult for classical methods such as the Gaussian kernel, but can be handled effectively by quantum kernel methods.

\begin{center}\
\begin{figure*}[!h]
    \centering
    \includegraphics[keepaspectratio, scale=0.7]{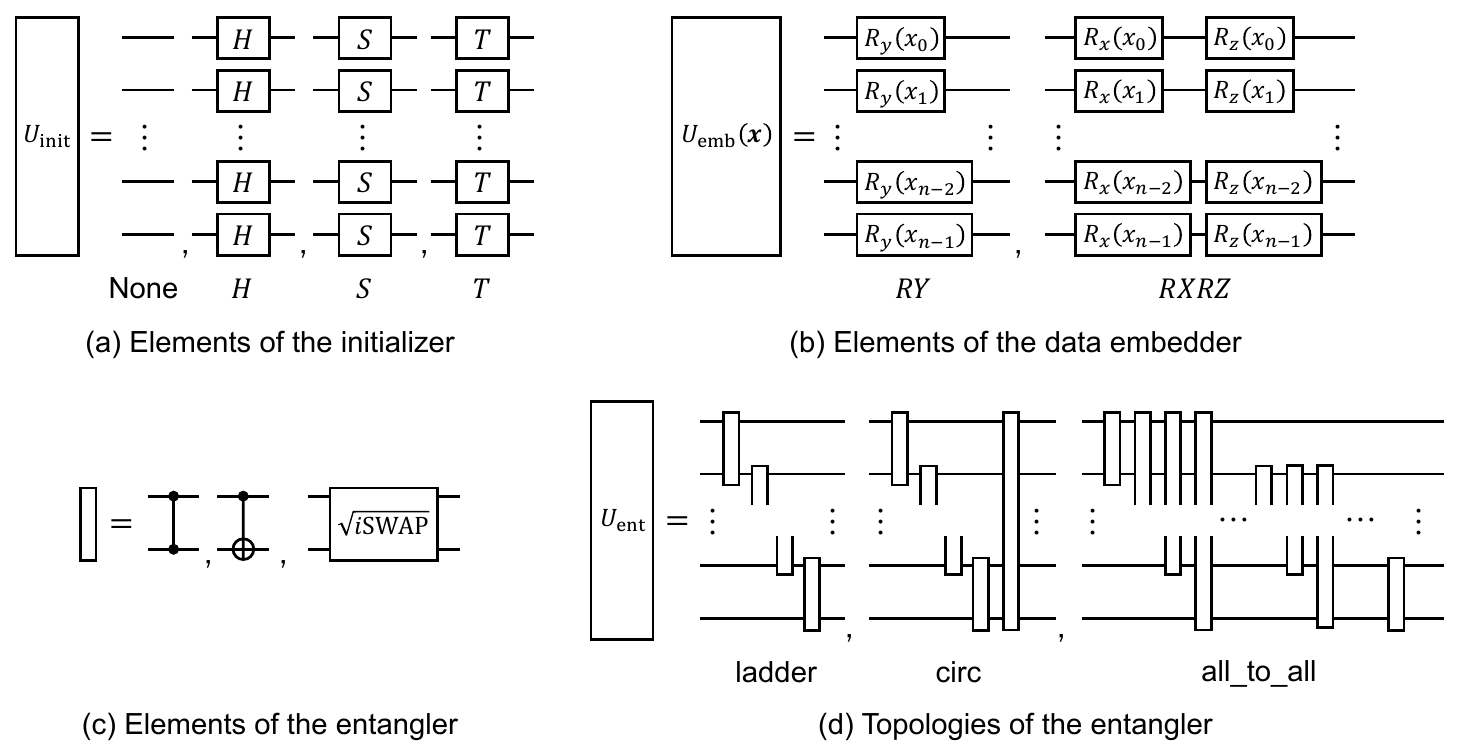}
    \caption{Elements of the quantum circuit}
    \label{fig:qc_parts}
\end{figure*}
\end{center}\

\vspace{8pt}

\begin{center}\
\begin{figure*}[!h]
    \centering
    \includegraphics[keepaspectratio, scale=0.8]{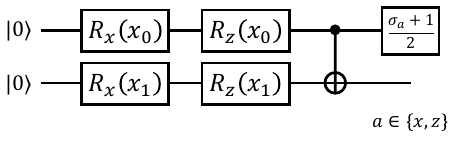}
    \caption{Quantum circuit of the data generator used in Sec.~\ref{sec:results_solosumo}}
    \label{fig:qc_datagen}
\end{figure*}
\end{center}\

\section{Proof of Lemma~\ref{lem:deriv}}\label{app:prf}

    For a given differentiable scalar-valued function $f(\mathbf{A})$ of matrix $\mathbf{A}$, it should be noted that
    \begin{eqnarray}
        \frac{d f}{d z} = \sum_{kl} \frac{\partial f}{\partial A_{kl}} \frac{\partial A_{kl}}{\partial z} = {\rm Tr} \left[ \left[ \frac{\partial f}{\partial \mathbf{A}} \right]^T \frac{\partial \mathbf{A}}{\partial z} \right].
    \end{eqnarray}
    Furthermore, if matrix $\mathbf{S}$ is symmetric, we derive 
    \begin{eqnarray}
        \frac{\partial \mathbf{S}}{\partial S_{ij}} = J^{ij} + J^{ji} - J^{ij}J^{ij},
    \end{eqnarray}
    where $J^{ij}$ denotes a single-entry matrix. Thus, for a given scalar function $f(\mathbf{S})$, we derive 
    \begin{eqnarray}
        \frac{d f}{d \mathbf{S}} = \left[ \frac{\partial f}{\partial \mathbf{S}} \right] + \left[ \frac{\partial f}{\partial \mathbf{S}} \right]^T - {\rm diag} \left[ \frac{\partial f}{\partial \mathbf{S}} \right]. \label{eq:d_f_sym}
    \end{eqnarray}
    In particular, for matrix $\mathbf{A}$ and symmetric matrix $\mathbf{S}$, Eq.~\eqref{eq:d_f_sym} results in 
    \begin{eqnarray}
        \frac{\partial {\rm Tr}[ \mathbf{AS} ]}{\partial \mathbf{S}} = \mathbf{A} + \mathbf{A}^T - (\mathbf{A} \circ \mathbf{I}).
    \end{eqnarray}

    Using the above equations, we can calculate the following:
    \begin{strip}
    \rule[-1ex]{\columnwidth}{1pt}\rule[-1ex]{1pt}{1.5ex}
    \begin{eqnarray}
        \frac{\partial}{\partial \theta} {\rm Tr} \left[ \mathbf{K}_X \mathbf{K}_Y \right] 
        &=& 
        {\rm Tr} \Biggl[ \left( \frac{\partial {\rm Tr}[\mathbf{K}_X \mathbf{K}_Y]}{\partial \mathbf{K}_X} \right)^T \frac{\partial \mathbf{K}_X}{\partial \theta} %
         + \left( \frac{{\partial \rm Tr}[\mathbf{K}_X \mathbf{K}_Y]}{\partial \mathbf{K}_Y} \right)^T \frac{\partial \mathbf{K}_Y}{\partial \theta} \Biggr] \\
        &=& 
        {\rm Tr} \left[ \left( \frac{\partial {\rm Tr}[\mathbf{K}_X \mathbf{K}_Y]}{\partial \mathbf{K}_X} \right)^T \frac{\partial \mathbf{K}_X}{\partial \theta} \right] \\
        &=& 
        {\rm Tr} \left[ \left( 2 \mathbf{K}_Y - \mathbf{K}_Y \circ \mathbf{I} \right) \partial_{\theta} \mathbf{K}_X \right], \label{eq:d_tr_AB} \\
        \frac{\partial}{\partial \theta} {\rm Tr} \left[ \mathbf{K}_X^2 \right] &=& {\rm Tr} \left[ \left( \frac{\partial {\rm Tr}[\mathbf{K}_X^2]}{\partial \mathbf{K}_X} \right)^T \frac{\partial \mathbf{K}_X}{\partial \theta} \right] \\
        &=& 
        {\rm Tr} \left[ \left( 4 \mathbf{K}_X - 2\mathbf{K}_X \circ \mathbf{I} \right) \partial_{\theta} \mathbf{K}_X \right]. \label{eq:d_tr_A2}
    \end{eqnarray}
    Therefore, we derive that
    \begin{eqnarray}
        \frac{\partial f}{\partial \theta} 
        &=& 
        - \frac{ \partial_{\theta} {\rm Tr}[ \mathbf{K}_X \mathbf{K}_Y ]}{{\rm Tr}[ \mathbf{K}_X \mathbf{K}_Y ]} %
         + \frac{ \partial_{\theta} {\rm Tr}[ \mathbf{K}_X^2 ]}{2{\rm Tr}[ \mathbf{K}_X^2 ]} + \frac{ \partial_{\theta} {\rm Tr}[ \mathbf{K}_Y^2 ]}{2{\rm Tr}[ \mathbf{K}_Y^2 ]} \\
        &=& 
        - \frac{{\rm Tr} \left[ \left( 2 \mathbf{K}_Y - \mathbf{K}_Y \circ \mathbf{I} \right) \partial_{\theta} \mathbf{K}_X \right]}{{\rm Tr}[ \mathbf{K}_X \mathbf{K}_Y ]} %
         + \frac{{\rm Tr} \left[ \left( 2 \mathbf{K}_X - \mathbf{K}_X \circ \mathbf{I} \right) \partial_{\theta} \mathbf{K}_X \right]}{{\rm Tr}[ \mathbf{K}_X^2 ]}.
    \end{eqnarray}
    \hfill\rule[1ex]{1pt}{1.5ex}\rule[2.3ex]{\columnwidth}{1pt}
    \end{strip}
    The case of $\partial_{\phi} f$ can be derived similarly.

\section{Application to biological data with gold standard network}\label{app:sachs}

To verify the applicability of the qPC algorithm, we systematically investigate the performance of the PC and qPC algorithms for the gene expression data, where the underlying causal relation is characterized by the gold standard network (\cite{sachs2005causal}).
We used the dataset from (\cite{sachsurl}).
The data describe the signal processing in proteins and phospholipids within human cells, comprising 11 variables.
We compared the inference results with the gold standard network using ROC curves to estimate how well the causal discovery algorithms could reconstruct the underlying causal relations from the data.
The ROC curves for the three algorithms with different sample sizes are shown in Fig. \ref{figure:sachs}.
All algorithms exhibit an improvement in reconstructing the gold standard network as the sample size increases.
We see no significant difference in the performance of the three methods.

\end{document}